\documentclass{svproc}

\usepackage{url}

\usepackage{multicol}
\usepackage{footmisc}
\usepackage[]{graphicx}
\usepackage{mathrsfs, amsmath, amssymb,amsfonts}
\usepackage{verbatim}
\usepackage[english]{babel}

\numberwithin{equation}{section}

\newtheorem{assum}{Assumption}

\newtheorem{thm}{Theorem}[section]\numberwithin{thm}{section}
\newtheorem{rem}{Remark}[section] \numberwithin{rem}{section}
\newtheorem{defn}{Definition}[section]\numberwithin{defn}{section}
\newtheorem{lem}{Lemma}[section]\numberwithin{lem}{section}
\newtheorem{cor}{Corollary}[section]\numberwithin{cor}{section}

\newcommand{\Dom}{\mathcal{D}}
\newcommand{\Ran}{\mathcal{R}}
\newcommand{\Ker}{\mathrm{Ker}}

\newcommand{\bO}{{\Gamma}}
\newcommand{\Ad}{A_0}
\newcommand{\An}{A_1}
\newcommand{\Abb}{{A_{\beta_0, \beta_1}}}
\newcommand{\Gd}{\Gamma_{\!0}}
\newcommand{\Gn}{\Gamma_{\!1}}
\newcommand{\Complex}{\mathbb{C}}
\newcommand{\Real}{\mathbb{R}}
\newcommand{\I}{{\mathrm{Im}\,}}

\begin{document}

\mainmatter              

\title{Linear Operators and Operator Functions Associated with Spectral Boundary Value Problems}
\titlerunning{Spectral Boundary Value Problems} 
\author{Vladimir Ryzhov}

\institute{
} 





\maketitle

\begin{abstract}

The paper develops a theory of
 spectral boundary value problems
from the perspective of
   general theory of linear operators in Hilbert spaces.
An abstract form of spectral boundary value problem
  with generalized boundary conditions
   is suggested and results
    on its solvability complemented by representations
     of weak and strong solutions are obtained.
 Existence of a closed
  linear operator defined by a given
   boundary condition and description
    of its domain are studied in detail.
 These questions are addressed on the basis
  of Krein's resolvent
   formula derived from the explicit
    representations of solutions also obtained here. 
Usual resolvent identities
 for two operators associated with two different
  boundary conditions are written in terms of
    the so called M-function.
Abstract considerations are complemented
 by illustrative examples taken from the
 theory of partial differential operators.
 Other applications to
   boundary value problems
 of analysis and mathematical physics
    are outlined.
\keywords{Spectral boundary value problem, singular perturbations,
   M-function, Krein's resolvent formula, linear operators, open systems theory}
\end{abstract}


\smallskip


\section{Introduction}\label{INTRO}

Close relationships between studies of
  boundary value problems and the linear
   operator theory are well known
    to specialists in both  disciplines.
 As an example, one can only mention numerous attempts
   to translate properties of solutions
    to boundary value problems
     into the operator-theoretic language
      that culminated in the development of an
      important branch of contemporary mathematics,
      the  interpolation theory
      of linear operators and
      scales of Banach and Hilbert
      spaces~\cite{Berez}, \cite{KPS}, \cite{LiMa}.
 Another achievement of
  the abstract operator theory in relation
  to boundary problems arising in applications is
  the extension theory of symmetric operators.
With its origin in quantum mechanics
 and operating in the setting of Hilbert spaces,
   the extension theory suggests a convenient model
     of boundary value problems rooted in Hilbert
       space operator theory~\cite{AhGl}.
Although the abstract approach often
 turns out to be too generic and therefore additional
  considerations are required
   to complete the study,
   the extension theory of symmetric operators continues to be an important
    and widely used tool
     in the studies of boundary value problems in abstract settings.
 During past decades
  it was substantially enhanced and enriched
   by various applications
    to the operator theory itself, to the classical
     and functional analysis, and to the mathematical physics.
The long list of publications~\cite{AlpBe}, \cite{BehrLan}, \cite{Birman}, \cite{BGW}, \cite{BHMNW}, \cite{BMNW}, \cite{Bruk}, \cite{DHMS1}, 
\cite{DHMS2}, \cite{DHMS3}, \cite{DM1}, \cite{DM2}, \cite{DM3}, \cite{GMT}, \cite{GM1}, \cite{GM2}, \cite{Grubb}, 
\cite{Grubb3}, \cite{HKS}, \cite{Kochub1}, \cite{Kochub2}, \cite{LangText}, \cite{Mal}, \cite{MalMog1}, \cite{MalMog2}, 
\cite{MalMog3}, \cite{Posil1}, \cite{Posil2}, \cite{Posil3}, \cite{Posil4}, \cite{Posil5}, \cite{Vi}
reflects only a small portion of the sheer amount of ongoing studies in
 the field of extensions theory of symmetric operators.
%
%


The present paper offers an operator-theoretic treatment
 of boundary value problems.
The main topic under discussion is the existence of Hilbert space operators
   corresponding to abstract linear boundary value problems
    defined by suitably generalized boundary conditions.
As is well known, many  applications of the partial differential
 equations theory entail problem statements characterized by certain types of
   formally written boundary conditions.
In the case of second order partial differential equations on bounded or 
 unbounded domains
  such conditions are usually rendered
   in terms of linear combinations of boundary values of solutions
     and traces of their derivatives evaluated on the domain boundary.
Typical examples include the Laplacian in a domain of Euclidian space
  with Dirichlet, Neuman, or Robin boundary conditions.
  Depending on the nature of these conditions, the resulting problem
   may or may not give rise to a closed linear operator in a Hilbert or Banach space.
 If such an associated operator exists, then
  the study is effectively reduced to the analysis
    of its properties.
  The paper describes a wide class of boundary conditions
   that determine a closed linear operator in a Hilbert space and studies
   its spectral characteristics in the general setting.
In a sense, the goal pursued here is opposite
 to the treatment by G.~Grubb~\cite{Grubb} where the existence of
  boundary conditions corresponding to a given closed realization of an
    elliptic operator is investigated.
%
%


%
%
An essential part of present research
 is the formulation of generic
  linear boundary value problems in the language of
    Hilbert space operator theory.
%
%
Within this framework, boundary conditions are
  defined by two parameters, 
     two closed linear operators acting in the ``boundary space.''
Reiteration of the material developed earlier in~\cite{Ryz1}, \cite{Ryz2}
  is followed by a more detailed
   inquiry into the properties of solutions, which in turn leads to
    Krein's resolvent formula and usual resolvent identities
      for closed operators acting on the ``main space'' corresponding
       to various boundary conditions.
Spectral properties of these operators
  are described in terms of the so called
     M-function.\footnote{
Values of all M\nobreakdash-functions under consideration are closed
 linear operators acting in Hilbert spaces, with one exception of the
 example given in Section~7 where M\nobreakdash-function is a matrix
 function.
Sometimes the term ``M\nobreakdash-operators'' is used
 in order to stress
  the operator theoretic nature of M\nobreakdash-function~\cite{AmP}.}
Several examples offered throughout the text illustrate the main ideas.
%


 Ongoing study of M\nobreakdash-functions, also known
  as
   $m$-functions, $Q$\nobreakdash-functions,
    Weyl-Titchmarsh functions,
     Steklov-Poincar\'e operator,
      Dirichlet-to-Neumann
       maps, transfer functions, etc.
       forms a significant part
        of the contemporary boundary value problems studies.
 The M-functions theory originates
  in the concept of \textit{m}\nobreakdash-function
    for singular Sturm\nobreakdash-Li\-o\-u\-ville differential
      equations~\cite{Titch}.
Since then the notion of M\nobreakdash-functions
 has been generalized to other settings 
   and followed by deep results
     on M-function properties and applications.
 We only mention a few relevant papers
  concerning topics in scattering theory~\cite{Pavlov1}, \cite{Uhlmann},
  Schr\"odinger and
  Sturm\nobreakdash-Liuoville operators theory~\cite{AmP},
  inverse problems~\cite{Isakov}, \cite{SU},
  the spectral asymptotic~\cite{Fried}, \cite{Safar},
  extensions of symmetric operators and adjoint
  pairs~\cite{BGW}, \cite{BHMNW}, \cite{BMNW}, \cite{DHMS1}, \cite{DHMS2}, \cite{DHMS3}, \cite{DM1}, \cite{DM2}, 
\cite{LangText}, \cite{MalMog1}, \cite{MalMog2}, \cite{MalMog3},
  numerous studies on partial differential operators
  including operators in
  non\nobreakdash-smooth
  domains~\cite{AlpBe}, \cite{GLMZ}, \cite{GMT}, \cite{GM1}, \cite{GM2}, \cite{GMZ1}, \cite{GMZ2},
  the numerical spectral analysis~\cite{BM}, \cite{Marl},
  singular perturbations~\cite{Posil1}, \cite{Posil2}, \cite{Posil3},
  and the linear systems theory~\cite{Ryz1}, \cite{Ryz2}.
 In the present paper's context
 M\nobreakdash-functions
  are realized
   as operator-functions with values in the set of closed linear
    operators acting in the ``boundary space,''
 a Hilbert space associated with the ``boundary."
%
%


Operator theoretic parts of the present work have some
 overlaps with
 the extensions theory of linear symmetric operators and
 relations in Hilbert and Krein spaces based on the notion of
  so called boundary triplets~\cite{Bruk}, \cite{Kochub1}
  and their 
    generalizations.
 This approach relies on  on properties of the abstract Green's formula
  that involves a linear symmetric operator (a linear symmetric relation in general case) and two
   linear boundary maps into the ``boundary space.''
The theory of boundary triplests is one of the most generic treatements of general boundary relations  
 available today. 
%
  The interested reader is referred to the original papers~\cite{DHMS1}, \cite{DHMS2}, \cite{DHMS3}, \cite{DM1}, \cite{DM2}, 
\cite{DM3}, \cite{Mal}, \cite{MalMog1}, \cite{MalMog2}, \cite{MalMog3}
    where further references can be found.
Another successful  approach
 to the extension theory of symmetric operators
  was elaborated by A.~Posilicano in works~\cite{Posil1}, \cite{Posil2}, \cite{Posil3}, \cite{Posil4}, \cite{Posil5}.
 It is rooted
  in a close  relationship between singular
   perturbations of elliptic differential operators and the extensions
    theory~\cite{AGHH}, \cite{AlbevKur}.
In comparison with these studies, the present study follows the line of reasoning found in 
\cite{Ryz1}, \cite{Ryz2}.
 The ideas expounded below are
 inspired by the Birman-Krein-Vishik method
  of extensions of positive operators in Hilbert
   space~\cite{Birman}, \cite{Krein2}, \cite{Vi} (see also~\cite{Grubb}, \cite{Grubb3} and \cite{AS}),
    the Weyl decomposition~\cite{Weyl},
     the open systems theory~\cite{Liv},
      and  the theory of linear systems
       with boundary control~\cite{Staff}.
 As a result, the framework in this paper
   is not centered around symmetric operators and 
    does not involve any notions
     specific to the extensions theory.
It is built from the first principles concerning
  linear operators and their domains, as well
   as properties of linear sets in Hilbert spaces.
 The linear systems theory conveniently provides
   an adequate language to
    communicate the underpinnings of this approach.
 With some risk of oversimplification,
  the last part of Section~2 explains these ideas
   in more depth by connecting them to the objects 
    of systems theory.
%
%
%
%


The paper's treatment of boundary value problems from the abstract
  point of view opens a possibility to consider classical and non-classical 
   applications from a uniform perspective.  
 As an example, it turns out that
  obtained results offer a straightforward interpretation
   of boundary value problems
    when the ``boundary'' does not exist a priori
     and has to be constructed artificially.
 This type of problems has been well studied
  in the literature and is
   usually referred to as singular
    perturbations of differential operators
    characterized by perturbations supported
     on the sets of Lebesque measure zero
      (often called the null sets).
 The well known quantum mechanical model of point
  interactions~\cite{AGHH}, \cite{AlbevKur} and the study of
   more general
    Schr\"odinger operators with potentials
     supported by null sets~\cite{AFHKL} are
      typical examples.
This fact underlines close connections
 of the present material
  to the papers~\cite{Posil1}, \cite{Posil2}, \cite{Posil3}, \cite{Posil4}, \cite{Posil5}
   devoted to the study of extensions
      of symmetric operators and singular perturbations.
 In the field of linear systems theory singular perturbations
  represent the procedure of ``channels opening''
   connecting an initially closed systems to its
    environment~\cite{Liv}.
 From this point of view, the M\nobreakdash-function is
  naturally identified with the
   transfer function of the resulting open system
    interacting with its environment
     by means of these channels.
Operator theoretic treatment also
 illuminates ideas behind the so-called ``Dirichlet
   decoupling''~\cite{DeSi} also known
    as ``Glazman's splitting procedure''~\cite{Glaz}
     and establishes connections to the analog
      of Weyl-Titchmarsh function of multidimensional
       Schr\"odinger operator~\cite{AmP}, \cite{Ryz2}.
It appears relevant to other problems of mathematical physics, e.~g.
 the exterior complex scaling in the
  theory of resonances~\cite{SiExtSc} and the R-matrix method
   well known in nuclear physics~\cite{LT}.
Some of  these applications are discussed in Section~7
 and in the last part of Section~2
  where relevant bibliographical references can be found;
   these ideas are the topics of further research.
%


The approach to spectral boundary value problems adopted in the paper
 has certain limitations.
One of them is the assumption of selfadjointness and bounded
 invertibility of
  the ``main'' operator (denoted $\Ad$ throughout)
   acting in the Hilbert space~$H$.
The requirement of bounded invertibility of $\Ad$ can be
 weakened to the condition~$\Ker(\Ad+ cI) = \{0\}$ for some $c \in \Real$, but the selfadjointess is essential.
Nevertheless, the schema can be extended to the case~$\Ad \neq\Ad^*$, but only  at the expense
  of introducing the so called dual pairs~\cite{MalMog1}, \cite{MalMog2}, \cite{MalMog3}
   (see also~\cite{BGW}, \cite{BHMNW}, \cite{BMNW}), which makes the study much more involved.
Another limitation is the preference to work with  linear operators, rather than with linear
 relations (multivalued operators) which appears to be the recent trend in the literature,
  see especially~\cite{DHMS2}, \cite{DHMS3}.
The language of single valued operators 
stands more in line with the classical approach of operator theory and is preferred here. 
One more requirement is that the operator~$\Ad$ must be unbounded so that the range of~$\Ad^{-1}$ is dense in~$H$.
Fortunately, all these restrictions do not impede the study of the main
 question addressed in the paper, that is,
  the description of operators corresponding to boundary value problems
   defined in terms of boundary conditions.
 %
%
%


Let us now briefly overview the paper's structure.
Section~\ref{DNMAP} offers an accessible introduction  into
 the setting of boundary problems and M\nobreakdash-functions.
It serves as a guideline for the topics discussed later
 and provides a concise exposition of
 the operator theoretic framework of spectral
 boundary value problems independent of the symmetric operators theory.
The main example is the well known spectral
 problem for Dirichlet Laplacian in a smooth domain in~$\Real^n$,
 $n \ge 3$.
 An adequate language for study of this operator
 is the language of Green's functions, integral
 equations and layer potentials.
When necessary, relevant results are freely borrowed from the standard
 references~\cite{Agran}, \cite{Maz}, \cite{ML}.
By this example all essential
 ingredients of the following exposition
   are explicitly formulated and finally compiled
     in a short catalog.
Relationships to the
 extension theory of symmetric operators, Krein's
   resolvent formula, and resolvent identities
    are also discussed.
Since the linear systems theory plays an important role
 for the approach employed in this work,
  a brief explanation of the principal ideas of this 
  theory is provided for reader's convenience.
 At the end of section
  other cases of partial differential operators that can be treated in 
   a similar fashion
    are mentioned.
%
%


%
%
Section~\ref{ABST} develops the machinery required
 for the purposes of the paper.
 The main objective here is to formulate notions useful for
  the study of spectral boundary value problems and associated
  M\nobreakdash-functions given in terms of their basic
   underlying objects.
 Such objects are two Hilbert spaces and
    three closed linear operators satisfying certain
     compatibility conditions.
 The solvability theorem is proven and ensuing
   definitions of weak and strong solutions are discussed.
 The section concludes with alternative descriptions
  of M\nobreakdash-functions and some
  comments regarding their properties.
%
%


%
%
Spectral boundary value problems with general boundary conditions
  are  investigated in Section~\ref{BOUNDC}.
 After the problem statement
  the solvability theorem is proven
   and  expressions for the
    solutions corresponding to
     various boundary conditions
      are obtained.
 The last part of the section explores a general definition of
  M\nobreakdash-functions associated with two different
   boundary conditions.
%
%
%


Section~\ref{EXTENT} is the main contribution of the paper.
 We discuss the existence
  of closed linear operators corresponding to spectral
   boundary value problems and
   subsequent study of their properties.
 Formal expressions for the resolvents
   and parameters of ``boundary conditions''
    are derived from the general representation of solutions
     obtained in the previous section.
 These  expressions are rigorously justified
   alongside with the study of
    spectral properties of respective operators
     and detailed descriptions of their domains.
 Relations to the extension theory
  of symmetric operators are also explained.
 The section closes with a brief digression into the
 original Birman\nobreakdash-Krein\nobreakdash-Vishik
 theory~\cite{Birman}, \cite{Krein2}, \cite{Vi} and
 remarks on its connections to the present study.
%
%


 Section~\ref{CAYLEY} offers a sketch of
  scattering theory for operators
   associated with boundary value problems.
Simultaneously,
  by virtue of the paper's approach, all arguments of this section remain valid
    for singular perturbations of partial differential operators
     by ``potentials'' concentrated on null sets.
Form the operator-theoretic point of view, the primary interest here
is the link between the
 boundary value problems theory
 and the functional model of nonselfadjoint operators
 established by means of Cayley transform
 applied to the M\nobreakdash-function.
It is shown that the ideas of papers~\cite{Na2}, \cite{Na3} devoted to the
 functional model based approach to the scattering theory
 are easily adopted and are fully applicable 
 for the comprehensive development of scattering theory
 of linear boundary value problems, selfadjoint and
 nonselfadjoint alike.
 Section~\ref{CAYLEY} can be seen as a groundwork for the future study in this direction.
%
%
%
%


The last section is
 an illustration of the boundary value problem
 technique discussed in the paper in application to
 singular perturbations of multidimensional
 differential operators.
 A simple example of the
 quantum mechanical model for a
 finite number of point
 interactions in $L^2(\Real^3)$
 (see~\cite{AGHH}, \cite{AlbevKur}) is studied.
 A familiar interpretation in the form of
 Schr\"odinger operator with $\delta$\nobreakdash-potentials
 is given and additional comments regarding singular perturbations
 concentrated on the null sets are supplied.
 Reported results are by no means new; most of them can be easily
 found in the relevant literature cited in the text.
 The objective of this section is to demonstrate how the abstract
 schema presented in earlier chapters can be put to practice
 for the study of particular cases of multidimensional differential operators.

\smallskip

The concept of this paper took shape during my visit
 to Cardiff University in April 2008.
I would like to express my sincere gratitude to
  Prof.~Marco~Marletta and  Prof.~ Malcolm~Brown
 for the invitation, hospitality, and inspiring discussions.
I am also grateful  to Prof.~Serguei~Naboko for his interest to the
 work, multiple discussions concerning its style,
 and continuous encouragement.
Last but not least, I am indebted  to an anonymous referee for his/her thorough 
examination of the manuscript and the list of relevant research  papers published after the
present work  was completed, see~\cite{A01}, \cite{A02}, \cite{A03}, \cite{A04}, 
\cite{A05}, \cite{A06}, \cite{A07}, \cite{A08}, \cite{A09}, \cite{A10}, 
\cite{A11}, \cite{A12}. 
The comments that I was able to act on definitely improved the text;
as for the others, I am hopeful to be able to address them 
in future publications.  

 %

\smallskip

%

This work is dedicated to the memory of Boris Pavlov (1936 -- 2016),
the role model of my academic career.


\begin{paragraph}{Notation}
Symbols~$\Real$, $\Complex$, $\I (z)$ stand for the real axis, the
complex plane, and the imaginary part of a complex number~$z\in
\Complex$, respectively.
The upper and lower half planes are the open sets~$\Complex_\pm :=
\{ z \in \Complex \; | \;\pm\I (z)
> 0 \}$.
If $A$ is a  linear operator on a separable Hilbert space~$H$, the
domain, range and null set of~$A$ are denoted $\Dom(A)$, $\Ran(A)$,
and $\Ker(A)$, respectively.
 For two separable Hilbert spaces~$H_1$ and~$H_2$ the
notation~$A : H_1 \to H_2 $ is used for a bounded linear
operator~$A$ defined everywhere in~$H_1$ with the range in the
space~$H_2$.
The symbol~$\rho(A)$ is used for the resolvent set of~$A$.
For a Hilbert space~$H$ the term \textit{subspace} always denotes a
closed linear set in~$H$.
The closure of operators and sets is denoted by the horizonal bar
 over the corresponding symbol.
All Hilbert spaces are assumed separable.

\end{paragraph}


\section{Boundary Value Problems by Example}\label{DNMAP}

In this introductory section we recall the classical example of the
boundary value problem and  M\nobreakdash-function associated
with the Dirichlet Laplacian in a simply connected bounded domain
with smooth boundary in the Euclidian space.
 The purpose of this exposition is twofold.
First, it
 reminds the reader of
 the concept of M\nobreakdash-functions, and
 secondly it brings together
 facts that serve as a
 foundation for the general approach
 developed further.
 The \textit{italic typeface} is used to highlight
  those observations which are essential for
   the investigations of the present paper.
Results cited below
 hold true under much  weaker assumptions,
  e.~g. for elliptic differential
   operators on non-smooth domains
    including Lipschitz subdomains of Riemannian manifolds,
     see~\cite{Kenig}, \cite{ML}, \cite{Mitrea} and references therein.
For further details the reader is referred to many
 expositions of the boundary integral equations method in application
 to boundary value problems for elliptic equations and systems, see
 \cite{Agran}, \cite{Maz}, \cite{ML} for relevant references.
%
%
%
%


\paragraph*{Dirichlet problem}
Let $\Omega \subset \Real^n$, $n\geq 3$ be a bounded simply
connected domain with
 $C^{1,1}$-boundary~$\bO$.
 The Laplace
 differential expression $ \Delta = \sum_i\frac{\partial^2}{\partial x^2_i}$
 defined on smooth functions in~$\Omega$
 generates the  Dirichlet Laplacian $ \Delta_D$
 in~$L^2(\Omega)$.
The domain of  $\Ad := -\Delta_D$ consists of functions from the
Sobolev class~$H^2(\Omega)$ with null traces on $\bO$.
 \textit{The operator~$\Ad$ is selfadjoint
 and boundedly invertible in $L^2(\Omega)$}.
%
%
%
%


\paragraph*{Harmonic functions and operator of harmonic continuation}
 Let $\gamma_0$ be the trace operator that maps continuous functions $u$
 defined in the closure~$\overline\Omega$ of $\Omega$ into
 their traces on the boundary, $\gamma_0 : u\mapsto
 \left.  u\right|_{\bO}$.
\textit{It follows from the definition of $-\Delta_D$ that $\gamma_0
\Ad^{-1} = 0$}.
 For $\varphi \in C(\bO)$
 denote $h_\varphi$ the solution
 of the Dirichlet problem in~$\Omega$:
 \[
  \Delta u = 0, \qquad \gamma_0 u = \varphi, \quad \text{ where }\quad \varphi\in C(\bO)
 \]
 The operator~$\Pi :\varphi \mapsto h_\varphi$
 is bounded as a mapping from $L^2(\bO)$
 into $L^2(\Omega)$  and
 $\Ker(\Pi) = \{0\}$, see~\cite{ML}.
It is readily seen that $\Pi$ is
 the classical operator of harmonic continuation
 from the boundary $\bO$ into the domain~$\Omega$
   uniquely extended to the bounded linear map
    defined on the space~$L^2(\bO)$.
 The equality $\gamma_0 \Pi \varphi = \varphi$ continues to
 hold for $\varphi \in L^2(\bO)$ and moreover
 $ \Delta h_\varphi =  0$
 for $h_\varphi = \Pi \varphi $
 in the sense of distributions~\cite{ML}. 
Observe that the (unbounded and not closed) operator  $A : \Ad^{-1} f \dot{+} \Pi \varphi \mapsto f$, 
$f \in L^2(\Omega)$ , $\varphi \in L^2(\Gamma)$ 
is well defined since 
\textit{the domain of operator~$\Ad$ and the
 set~$\Ran(\Pi)$ do not have nontrivial common elements}, otherwise
 $\Ad$ would not be boundedly invertible:
\[
\exists \, \Ad^{-1} \Rightarrow \Dom (\Ad) \cap \Ran(\Pi) =\{0\}
\]
The same argument shows that $\Dom(\Ad)$ does not contain
 any nontrivial functions from $H^2( \Omega)$
 satisfying the homogenous equation $(-\Delta - zI) h =0 $  under the assumption $z\in
 \rho(\Ad)$.
Obviously, $\Ad$ is a restriction of $A$ to $\Dom(\Ad)$. Notice also that $A$ does not coincide 
with the ``maximal operator'' defined as an adjoint to the map $ u \mapsto - \Delta u $, 
where $u \in L^2 (\Omega)$ belongs to the class $C_0^\infty$  of infinitely differentiable functions 
 vanishing in the vicinity of boundary $\Gamma$.


\paragraph*{Adjoint of the harmonic continuation operator}
Let $G(x,y)$ be the Green's function of $\Ad = -\Delta_D$,
 so that $(\Ad^{-1} f)(x) = \int_\Omega G(x,y)f(y)\,dx$
 for $f\in L^2(\Omega)$,
 see~\cite{ML}.
 The kernel~$G(\cdot, \cdot)$ is symmetric and real-valued:
  $G(x,y) = G(y,x)$ and $\overline{G(x,y)} = G(x,y)$.
Denote by $d\sigma$ the normalized Lebesgue surface measure
 on~$\Gamma$.
 Then the operator~$\Pi$ can be expressed as an
 integral operator with Poisson kernel
\[
 \Pi : \varphi \mapsto   -
   \int_{\bO} \varphi(y) \frac{\partial}{\partial\nu_y}\,G(x,y) \, d \sigma_y
\]
where $\frac{\partial}{\partial\nu}$ is the derivative along the
outside pointing normal  at the boundary~$\bO$.
 For a smooth function $f$ in
  $\Omega$
\[
 (\Pi\varphi, f) =  - \int_\Omega \left(\int_{\bO} \varphi(y) \frac{\partial}{\partial\nu_y}\,G(x,y) \, d \sigma_y \right)
 \overline{f(x)}\, dx
\]
and due to Fubini's theorem and properties of~$G(\cdot, \cdot)$, 
\begin{align*}
 (\Pi\varphi, f)  & =
 - \int_{\bO} \varphi(y) \, \frac{\partial}{\partial\nu_y} \left( \int_\Omega
 \overline{f(x)}\, \,G(x,y) \, dx\right)  \, d \sigma_y  \\ 
&  =  - \left \langle
   \varphi,
 \frac{\partial}{\partial\nu} \left( \int_\Omega
 G(x,\cdot) \, {f(x)} \, d x\right)
  \right\rangle 
\end{align*}
 where $\langle \cdot, \cdot\rangle$ denotes the inner product
 in~$L^2(\Gamma)$.
Since $G(x,y) = G(y,x)$ is the integral kernel of
 $\Ad^{-1}$,
 we obtain  \textit{the representation for~$\Pi^*$,  the adjoint
 of $\Pi$},
\[
 \Pi^* = \gamma_1 \Ad^{-1}
\]
where $\gamma_1 : u \mapsto  - \gamma_0  \frac{\partial
u}{\partial\nu} = -\left. \frac{\partial
 u}{\partial\nu}\right|_\Gamma$.
We will use the symbol~$\partial_\nu$ for the
 map~$u \mapsto \left. \frac{\partial
 u}{\partial\nu}\right|_\Gamma$, so that $\gamma_1  = -\partial_\nu$.


 %

\paragraph*{The spectral problem}
 The spectral Dirichlet boundary value problem for the differential 
expression ~$ \Delta  = \sum_i \frac{\partial^2}{\partial x_i^2}$ in~$\Omega$ 
 is defined  by the system of
  equations for $u \in \Dom(A) := \Dom(\Ad) \dot{+} \Ran (\Pi)$ 
 \begin{equation}\label{eqn:aux3}
 \left\{
 \begin{aligned}
  &&(A - zI) u = 0, \\
  && \gamma_0 u = \varphi
 \end{aligned}
 \right.
 \end{equation}
where $A : u \mapsto -\Delta u $,  $\varphi \in L^2(\bO)$, and the number~$z\in \Complex$ plays
the role of
 spectral parameter.
For $z\in \rho(\Ad)$ the distributional solution~$u_z^\varphi$ can
be obtained from the harmonic function $\Pi \varphi$ by the
 formula
 $u_z^\varphi =  (I- z\Ad^{-1})^{-1}\Pi\varphi$.
Indeed, since $(I- z\Ad^{-1})^{-1} =  I + z (\Ad -zI)^{-1}$
 and $A \Pi \varphi =0 $ in the  distributional sense,
 we have
 \[
 (A - zI) u_z^\varphi =  (A - zI)\left( \Pi \varphi +
 z (\Ad -zI)^{-1}\Pi\varphi \right) = - z \Pi\varphi +
 z\Pi\varphi =0
\]
 due to the identity~$(A - zI)(\Ad - zI)^{-1} = I$.
Therefore the vector~$u_z^\varphi$ is a solution to the equation~$(A
-z I)u =0$.
 Further, $\gamma_0 u_z^\varphi = \gamma_0 \Pi \varphi = \varphi$.
 Hence
 \textit{the vector~$u_z^{\varphi} = (I -
z\Ad^{-1})\Pi \varphi$ is a  solution to  the spectral
problem~(\ref{eqn:aux3}) for $\varphi\in L^2(\Gamma)$ and $z\in
\rho(\Ad)$}.


\paragraph*{Solution Operator and DN-Map}
For the spectral problem~(\ref{eqn:aux3}) with $\varphi \in
L^2(\bO)$ and $z\in \rho(\Ad)$ introduce the solution operator
 \begin{equation}\label{eqn:aux19}
 S_z : \varphi \mapsto (I - z\Ad^{-1})^{-1}\Pi\varphi
 \end{equation}
 Operator~$S_z$ is bounded as a mapping from $L^2(\Gamma)$
 into $L^2(\Omega)$.
 For $\varphi \in C^2(\Gamma)$
 the inclusion~$ S_z\varphi \in H^2(\Omega)$ holds
 and therefore
 the expression~$\gamma_1 S_z \varphi$ is well defined.
 The operator function $M(z)$ defined
 by
 \begin{equation} \label{eqn:aux17}
 M(z): \varphi \mapsto \gamma_1 S_z \varphi,\quad
  \varphi \in C^2(\Omega)
 \end{equation}
 is analytic in $z\in \rho(\Ad)$.
It
 is called the Dirichlet-to-Neumann map  (DN-map) or, more generally,
 the M\nobreakdash-function
  of $A =  -\Delta$ in
 the domain~$\Omega$.
By construction, $   - \partial_\nu u   = M(z)\left.
(u\right|_\Gamma)$
 for $u \in \Ker(A -zI)$
 as long as the
function~$\gamma_0 u = \left. u\right|_\Gamma$
 is sufficiently smooth on $\bO$.
In fact, it can be shown that values of so defined~$M(z)$, $z\in
\rho(\Ad)$ are closed operators
 acting in $L^2(\bO)$
 with the domain~$H^1(\bO)$, see~\cite{Taylor} and references therein.
%
%


The representation $S_z = (I - z\Ad^{-1})^{-1}\Pi $ and
equality~$\Pi^* = \gamma_1 \Ad^{-1}$ imply
 \begin{equation}\label{eqn:aux9}
(S_z)^* = \gamma_1 \Ad^{-1}(I - \bar z \Ad^{-1})^{-1} = \gamma_1
 (\Ad -\bar zI )^{-1}
 \end{equation}
Therefore~$S_z = [\gamma_1 (\Ad -\bar zI )^{-1}]^*$ and the
M\nobreakdash-function~$M(z)$ can be rewritten:
\[
M(z) = \gamma_1 [\gamma_1 (\Ad - \bar zI)^{-1}]^*
\]
In particular, $M(0) = \gamma_1 (\gamma_1 \Ad^{-1})^* = \gamma_1
\Pi$.
It can be shown (see~\cite{Taylor}) that
 \textit{ the operator~$ M(0) = \gamma_1 \Pi$ defined on the
 domain~$\Dom(M(0)) = H^1(\Gamma)$
 is selfadjoint in $L^2(\bO)$}.
Operator~$M(0)$ turns out to be a rather important object;
 it is convenient to use
  a special notation for it:
 \[
 \Lambda = M(0) = \gamma_1\Pi, \qquad \Dom(\Lambda) = H^1(\Gamma)
 \]
%

%
%
%
%
%
%

\paragraph*{Robin Boundary Conditions}
Let $\beta \in L^\infty(\bO)$ be a bounded function defined 
almost everywhere on the boundary $\Gamma$.
In what follows we also denote $\beta$ the bounded operator of multiplication
$\varphi \mapsto \beta \varphi$, $\varphi \in L_2(\bO)$ acting in the space $L_2(\bO)$.
Consider the boundary value problem
\begin{equation}\label{eqn:aux4}
 \left\{
 \begin{aligned}
  &&(A - zI) u = 0, \\
  &&   -\partial_\nu u + \left.\beta u \right|_\Gamma = \varphi
 \end{aligned}
 \right.
 \end{equation}
 with~$\varphi \in L^2(\bO)$.
In particular, for $\beta =0$ we recover the classical Neumann
problem for the Laplacian in~$\Omega$.
For nontrivial~$\beta$ the system~(\ref{eqn:aux4})
 is called the boundary problem of third type, or Robin problem.
 Assume  $z\in\rho(\Ad)$ and let $u_z^\varphi$
 be a smooth solution
 to the first equation, that is  $(A -zI)u_z^\varphi =0$.
 Because~$\gamma_1 u_z^\varphi = M(z)\gamma_0 u_z^\varphi$,
 the second equation for the trace  $\psi := \gamma_0 u_z^\varphi $ becomes
  $(\beta + M(z)) \psi = \varphi$.
Suppose the map~$(\beta + M(z))$ is boundedly invertible as an
 operator in $L^2(\bO)$.
 Then the boundary
 equation for $\psi$ can
 be solved explicitly:
 $\psi = (\beta + M(z))^{-1}\varphi$.
In turn,  the solution~$u_z^\varphi$ is recovered from its
trace~$\psi = \gamma_0 u_z^\varphi$ by the
 mapping~$S_z$:
 \begin{equation}\label{eqn:aux14}
 u_z^\varphi  = (I -z\Ad^{-1})^{-1}\Pi \gamma_0 u_z^\varphi =
 (I -z\Ad^{-1})^{-1}\Pi  ( \beta + M(z) )^{-1}\varphi,
 \end{equation}
where $z \in \rho(\Ad)$ is such that~$( \beta + M(z) )^{-1}$ exists.
Observe that application of~$\gamma_1$ to both sides of this
equality yields the expression for the map~$\varphi \mapsto\gamma_1
u^\varphi_z$, which  by analogy with the DN-map can be called the
Robin-to-Neumann map:
\[
 M_{RN}(z) = M(z)(\beta + M(z))^{-1}
\]
Similarly, application of~$\gamma_0$
 yields an expression for the Robin-to-Dirichlet map:
 \begin{equation}\label{eqn:aux10}
  M_{RD}(z) = (\beta + M(z))^{-1}
 \end{equation}
%

%
%

\paragraph*{Krein's resolvent formula and Hilbert resolvent identity}
 Equations~(\ref{eqn:aux4}) give rise to
 another boundary problem, namely the problem for an unknown function~$u$ in $\Omega$
 satisfying
\begin{equation}\label{eqn:aux5}
 \left\{
 \begin{aligned}
  &&(A - zI) u = f, \\
  &&   \gamma_1 u + \beta \gamma_0 u   = 0
 \end{aligned}
 \right.
 \end{equation}
 with~$f \in  L^2(\Omega)$, where $\gamma_1 u = -\partial_\nu u = - \frac{\partial u}{\partial \nu}|_\Gamma $
 and $\gamma_0 u = u|_\Gamma$.
It is customary to
 look for a solution to~(\ref{eqn:aux5}) in the form
 \begin{equation}\label{eqn:aux6}
 u_z^f  = (\Ad -z I)^{-1} f +
 S_z \psi =
(\Ad -z I)^{-1} f + (I -z \Ad^{-1})^{-1}\Pi \psi
 \end{equation}
  with $z\in \rho(\Ad)$ and some $\psi \in L^2(\bO)$ to be determined.
Since~$(A - zI)(\Ad -zI)^{-1} f  = f$ and $(A -zI)S_z \psi = 0 $,
 the first equation~(\ref{eqn:aux5}) is satisfied by~(\ref{eqn:aux6})
 automatically; therefore we only
 need to find $\psi \in L^2(\bO)$ such that
 (\ref{eqn:aux6}) obeys the boundary
 condition in~(\ref{eqn:aux5}).
 Applying $\gamma_0$ and $\gamma_1$ to~(\ref{eqn:aux6})
 we obtain
\[
  \begin{aligned}
  \gamma_0 u_z^f  & = \gamma_0 S_z \psi = \psi
   \\
   \gamma_1 u_z^f  &  = \gamma_1 (\Ad - zI)^{-1} f + \gamma_1 S_z
   \psi = \Pi^* (I - z\Ad^{-1})^{-1} f + M(z)\psi
  \end{aligned}
\]
Now
 the relation~$\Pi^* = \gamma_1 \Ad^{-1}$, properties of solution operator~$S_z$ and the
 definition of~$M(z)$, lead to the following equation for the unknown function~$\psi$
 \[
 0 = (\gamma_1 + \beta \gamma_0)u_z^f = \Pi^* (I - z\Ad^{-1})^{-1} f + (\beta + M(z))\psi
 \]
Again, assuming~$z\in \rho(\Ad)$ is such that $(\beta + M(z))$ is
boundedly
 invertible, the formula for~$\psi$ follows:
 \[
 \psi  = - (\beta + M(z))^{-1}\Pi^* (I - z\Ad^{-1})^{-1}
 f
 \]
Substitution into~(\ref{eqn:aux6}) yields the result
 \begin{equation}\label{eqn:aux7}
 u_z^f  = (\Ad -z I)^{-1} f -
 (I - z\Ad^{-1})^{-1} \Pi (\beta + M(z))^{-1}\Pi^* (I -
 z\Ad^{-1})^{-1} f
 \end{equation}
 This expression certainly requires some justification as the second summand
 need not be smooth and thereby the normal
 derivative~$- \partial_\nu u_z^f$
 that appears in the boundary condition
 may be undefined for some $f\in L^2(\Omega)$.
 But let us defer discussion of this difficulty to the main body of
 the paper and turn instead to the operator-theoretic interpretation of
 the equations~(\ref{eqn:aux5}) and their solution~(\ref{eqn:aux7}).

%
%

 The system~(\ref{eqn:aux5})
 represents a problem of finding a vector~$u$
 from the domain of operator~$\mathcal A_\beta$ defined as a restriction
 of~$A$ to the set of functions~$u \in L^2(\Omega)$
 satisfying the boundary condition~$(\gamma_1 + \beta \gamma_0) u =0$
 in some yet undefined sense.
 It is clear that~$\mathcal A_\beta$ also can be treated as
 an extension of the so-called \textit{minimal operator}
 defined as $A = - \Delta$ restricted to the set~$C^\infty_0(\Omega)$ of
 infinitely differentiable functions in $\Omega$ that vanish in
 some neighborhood of $\bO$ along with all their partial derivatives.
 Assuming for the sake of argument that
 each vector $u\in \Dom(\mathcal A_\beta)$ satisfies
 the condition~$(\gamma_1 + \beta \gamma_0) u =0$
 literally, that is the expression $(\gamma_1 + \beta \gamma_0)u$
 makes sense for each $u\in \Dom(\mathcal A_\beta)$,
 the problem~(\ref{eqn:aux5}) with $f\in  L^2(\Omega)$
 is the familiar resolvent equation~$(\mathcal A_\beta - zI) u =f$
 for the operator~$\mathcal A_\beta$.
 Therefore the solution~(\ref{eqn:aux7}) for $z\in \rho(\mathcal A_\beta)$
  coincides with
 $(\mathcal A_\beta -zI)^{-1} f$.
We see that the
 resolvents of $\Ad$ and $\mathcal A_\beta$
 for $z\in \rho(\Ad)\cap\rho(\mathcal A_\beta)$
 are related by the following identity
 commonly known as \textit{  Krein's resolvent formula}
\begin{equation}\label{eqn:aux8}
 (\mathcal A_\beta - zI)^{-1} = (\Ad - zI)^{-1}  -
(I - z\Ad^{-1})^{-1} \Pi (\beta + M(z))^{-1}\Pi^* (I -
  z\Ad^{-1})^{-1}
\end{equation}
Notice that the right hand side of~(\ref{eqn:aux8}),
 depends on $(\beta + M(z))^{-1}$
 which is exactly the M\nobreakdash-function~(\ref{eqn:aux10}).
Under assumption of bounded invertibility of $\beta + M(0)$ in
$L^2(\bO)$
 we have
 \begin{equation}\label{eqn:aux12}
  \mathcal A_\beta^{-1} = \Ad^{-1} - \Pi (\beta + M(0))^{-1}\Pi^*
 \end{equation}
 This expression shows in particular
 that while the difference of $\mathcal A_\beta$ and $\Ad$
 is only defined \textit{a priori} on the set of smooth
 functions~$u$
 vanishing on the boundary $\Gamma$ along with their first derivatives
 where $(\mathcal A_\beta - \Ad ) u  = 0 $,
 the difference of their inverses~$\mathcal A_\beta^{-1} - \Ad^{-1}$
 is a nontrivial bounded  operator in $L^2(\Omega)$.
 As a consequence,
 if $\beta = \beta^*$, then the operator ~$\mathcal A_\beta$ is
 selfadjoint as an inverse of a sum of two bounded selfadjoint
 operators.
 Moreover, the
  formula~(\ref{eqn:aux12}) can be successfully
  employed for the investigation into spectral
  properties of $\mathcal A_\beta$, as
  it reduces the boundary problem setting to
  the well-developed case of
  perturbation theory for bounded operators (cf.~\cite{Grubb}).

%
%
%
%

Krein's  formula~(\ref{eqn:aux8}) implies another
useful identity relating resolvents of $\Ad$
 and $\mathcal A_\beta$ to each other.
 According to
 the definition of solution operator~$S_z$ the identity
  $\gamma_0 (I -z\Ad^{-1})^{-1}\Pi = I$ holds for any~$z\in \rho(\Ad)$.
 Hence, application of $\gamma_0$ to both sides of~(\ref{eqn:aux8})
 leads to
\begin{equation}\label{eqn:aux11}
  \gamma_0 (\mathcal A_\beta - zI)^{-1} = (\beta + M(z))^{-1}\Pi^* (I -
  z\Ad^{-1})^{-1}
\end{equation}
 Krein's formula can now be rewritten in  the form
 \[
 (\mathcal A_\beta - zI)^{-1} -  (\Ad - zI)^{-1}  = -
 (I - z\Ad^{-1})^{-1} \Pi  \gamma_0 (\mathcal A_\beta - zI)^{-1}
 \]
 By substituting
 the adjoint
 of~$S_z = (I -z\Ad^{-1})^{-1}\Pi$
 from~(\ref{eqn:aux9})
 we obtain the following
 variant of \textit{ Hilbert resolvent identity} for $\Ad$  and
 $\mathcal A_\beta$ (cf.~\cite{GM1}, \cite{GM2})
 \begin{equation}\label{eqn:aux13}
 (\Ad - zI)^{-1} -
   (\mathcal A_\beta - zI)^{-1}
  = [\gamma_1  (\Ad -\bar zI)^{-1}]^*
  \gamma_0 (\mathcal A_\beta - zI)^{-1}, \quad
   z\in \rho(\Ad)\cap \rho(\mathcal A_\beta)
 \end{equation}
%
 %
%

%
%

Finally, notice that all considerations above are valid at least
 formally
 if the symbol~$\beta$ in the
 condition~(\ref{eqn:aux5}) represents a linear bounded operator
 acting on the Hilbert space~$L^2(\bO)$.

\paragraph*{Summary}
Observations of this section lay down a foundation for the
 study of boundary value problems and M\nobreakdash-functions presented in the
 paper.
 For further convenience, this preliminary discussion concludes by summing up
  properties of operators $\Ad$ and $\Pi$
  and their relationships to the
  boundary maps~$\gamma_0$, $\gamma_1$ that are relevant
  for our study.
\begin{itemize}
  \item Operator $\Ad^{-1}$ is bounded,  selfadjoint, and $\Ker(\Ad^{-1}) = \{0\}$
  \item
   Operator $\Pi$ is bounded and $\Ker(\Pi) = \{0\}$
   \item
   The intersection $\Dom(\Ad)\cap \Ran(\Pi) = \Ran(\Ad^{-1}) \cap
   \Ran(\Pi)
   $ is trivial
  \item
   The left inverse of $\Pi$ is the trace operator~$\gamma_0$
  restricted to~$\Ran(\Pi)$, that is
    $\gamma_0\Pi \varphi=
   \varphi $ for $\varphi \in L^2(\Gamma)$.
  \item
  The set~$\Dom(\Ad)= \Ran(\Ad^{-1})$ is included into
  the null space of~$\gamma_0$, so that $\gamma_0 \Ad^{-1}  =0$
  \item
   The adjoint operator of $\Pi$ is expressed in terms of $\gamma_1$ and $\Ad$ as $\Pi^* = \gamma_1 \Ad^{-1}$
  \item
   Operator $\Lambda = \gamma_1 \Pi$ is selfadjoint (and unbounded) in
   $L^2(\bO)$.
\end{itemize}
Further, the spectral boundary value problem~$(A -zI)u =0$,
 $\gamma_0 u = \varphi$, where~$A$ is an extension of~$\Ad$
 to the set~$\Dom(\Ad) \dot{+}\Ran(\Pi)$ defined as $A h =0$ for $h\in
 \Ran(\Pi)$,
 gives rise to
 the solution operator~$S_z$ and to the M\nobreakdash-function~$M(z)$, $z\in \rho(\Ad)$.
\begin{itemize}
  \item
   The solution operator has the form $S_z = (I
   -z\Ad^{-1})^{-1}\Pi$, $z\in \rho(\Ad)$
  \item
   The M-function is formally defined by the equality~$M(z) = \gamma_1
   S_z$, $z\in \rho(\Ad)$
\end{itemize}
 Finally, the boundary condition associated
 with the
 expression~$\gamma_1 + \beta\gamma_0$ where $\beta$ is a linear
 operator in $L^2(\Gamma)$ defines the Robin boundary value problem
 and the corresponding linear operator~$\mathcal A_\beta$.
\begin{itemize}
 \item
 The resolvents of $\mathcal A_\beta$
 of $\Ad$ are related
 by Krein's formula~(\ref{eqn:aux8})
 expressed in terms of
 M\nobreakdash-function~(\ref{eqn:aux10})
 \item
  Hilbert resolvent identity~(\ref{eqn:aux13}) holds.
\end{itemize}
%
%
%


\paragraph{The linear systems theory perspective}
As stated in Introduction, ideas underlying the operator theoretic framework
 employed for the paper's purpose
  are partially inspired by the approach
   to boundary value problems found
    in the linear systems theory.
 These ideas are best illustrated by considering the following variant of problem~(\ref{eqn:aux3})
 \begin{equation}\label{eqn:aux16}
 \left\{
 \begin{aligned}
  &&A   u = f, \\
  && \gamma_0 u = \varphi
 \end{aligned}
 \right.
 \end{equation}
where all participating objects are as in~(\ref{eqn:aux3}) and the vector~$f$ is an arbitrary function from $L^2(\Omega)$.
From the point of view of linear systems theory, equations~(\ref{eqn:aux16}) describe a linear system with the state space~$H = L^2(\Omega)$,
 the input-output space~$E = L^2(\Gamma)$ and the main operator~$A$.
Solutions to~(\ref{eqn:aux16})  are called ``internal states'' of the system and vectors~$\varphi \in L^2(\Gamma)$ are interpreted as the system's  input.
The system's output is defined by the operator~$\gamma_1$ that maps internal states of the system to elements of the input-output space~$E$.
%


%
When the input in (\ref{eqn:aux16}) is absent ($\varphi =0$),  the corresponding internal state is obviously $u^f = \Ad^{-1}f$.
This situation corresponds to the closed system, that is, the system that is isolated from the external influences modeled by inputs~$\varphi \in E$.
The closed system still has a nontrivial output given by~$\gamma_1 : u^f \mapsto \gamma_1 \Ad^{-1}f = \Pi^* f$.
Introduction of the non-zero input~$\varphi \in E$  in (\ref{eqn:aux16}) is a way to open the system to external influences.
As can be easily verified, the procedure of system opening results in an additional term in the expression for the state vectors, $u^{f,\varphi} = \Ad^{-1}f  + \Pi \varphi$.
The output of the system defined by the operator~$\gamma_1$ results in the mapping
 from internal states to outputs
 in the form $\gamma_1 : u^{f,\varphi} \mapsto \gamma_1 \Ad^{-1}f + \gamma_1 \Pi \varphi$.
 At this point
 we need to take into consideration the unboundedness
  of the trace operator~$\gamma_1$ and only choose inputs resulting in the outputs
   that belong to $E = L^2(\Gamma)$.
All such inputs (admissible inputs) therefore are
  functions~$\varphi \in L^2(\Gamma)$
   for which the harmonic continuations~$\Pi \varphi$
    into the domain~$\Omega$ possess normal derivatives with traces on~$\Gamma$ from the space~$L^2(\Gamma)$.
With an appropriate choice of inputs~$\varphi \in L^2(\Gamma)$, the system's output
 is determined by the map~$\gamma_1 : u^{f,\varphi} \mapsto \Pi^* f + \Lambda \varphi$, where we employed notation~$\Lambda =\gamma_1\Pi$
 introduced earlier and used the equality~$\gamma_1 \Ad^{-1} = \Pi^*$.
%


The restriction of admissible inputs to a smaller set in this example
 is dictated by the choice of output operator~$\gamma_1$ that
 cannot be defined on all attainable internal states~$\{u^{f,\varphi}\mid f \in H, \varphi \in E \}$ of the system.
Such a restriction however does not create any inconvenience.
Quite the opposite, this feature can be perceived as an advantage of the approach, because
 it allows  for the definition of inputs according
  to the particular problem at hand.\footnote{
 This situation is common in practical applications of the systems theory where the set of inputs
  is always subject to the real world limitations.
 For instance, it is clear that only smooth functions from~$L^2(\Gamma)$ can be realized in practice as the system's inputs.
Therefore the ability to choose input vectors freely conforms to the
 standard assumptions of systems theory.
}
Note also that an alternative approach consists of suitable alterations of outputs that do not change essential properties
 of the system under investigation,
 and at the same time
 widen the set of admissible inputs
 (see~\cite{BHMNW}, \cite{BMNW} in this regard for an example of ``regularization procedure''  applied to the system's output).
%


The internal states of the linear system described by equations~(\ref{eqn:aux16})  are therefore represented as the sum of two components,
 $u^{f,\varphi} = \Ad^{-1} f + \Pi \varphi$.
The first term is always a function from the domain of Dirichlet Laplacian, and the second term needs not be
 smooth and belong to the domain of $A = -\Delta$ at all.
It is a function from the range of the operator of harmonic continuation from the boundary, $\Pi :L^2(\Gamma) \to L^2(\Omega)$.
These two components are linearly independent
in the sense of equivalence
\{$u^{f,\varphi} = 0\} \Longleftrightarrow \{\Ad^{-1} f = 0,  \Pi \varphi = 0 \}$.
In other words, the internal states of the system are vectors from the direct sum~$\Dom(\Ad) \dot{+} \Ran(\Pi)$.
In  the language of linear systems theory the second summand  is associated with the set of controls imposed on the system.
 The equality $\Ker(\Pi) = \{0\}$ means
 that this set stands in a one-to-one
 correspondence with the set of all inputs.
Operator~$\Pi$ that maps inputs into controls is often called the \emph{control operator}.
%


For the ''spectral'' case of linear system described by equations (\ref{eqn:aux3})
 the system's input are again vectors $\varphi \in L^2(\Gamma)$
 and the internal state is determined by
 the solution operator~$ \varphi \mapsto S_z \varphi$ for $z\in\rho(\Ad) $, see~(\ref{eqn:aux19}).
Following the systems theory language, if the output is defined by means of operator~$\gamma_1$ as $\gamma_1 S_z\varphi$,
  then the M\nobreakdash-function~(\ref{eqn:aux17}) is nothing but the transfer
  function of this system that maps the input~$\varphi$  into
   the output~$\gamma_1 S_z \varphi$ (for suitable $\varphi \in E$).
The resolvent identity and formula~$\Pi^* = \gamma_1 \Ad^{-1}$ allow to rewrite (\ref{eqn:aux17}) as
 \begin{equation} \label{eqn:aux18}
 M(z) = \Lambda + z \Pi^* (I - z\Ad^{-1})^{-1}\Pi, \quad z \in \rho(\Ad)
 \end{equation}
This representation has important consequences.
%
%
%

First,
the function~(\ref{eqn:aux18}) is expressed  in terms
 of three linear operators, $\Ad^{-1}$, $\Pi$, and $\Lambda$, playing
  very specific and well defined roles in the description of linear system corresponding to~(\ref{eqn:aux3}).
Namely, many applications of
 the systems theory interpret  the spectral parameter~$z \in \Complex$ in~(\ref{eqn:aux3})
 as the frequency of oscillations taking place inside~$\Omega$.
Typical and well known examples are classical
  acoustic and electromagnetic waves existing in the domain~$\Omega$.
The operator~$\Lambda = M(0)$ then has
  the meaning of system's response at zero frequency,
   and can be interpreted as the operator of static reaction.
Since its independence on the spectral parameter
 it maps inputs directly to the outputs without applying any $z$-dependent (therefore, frequency dependent) transformations.
In the systems theory terms the
 operator~$\Lambda$ is usually called the \emph{feedthrough operator}.
Consequently, with a given input~$\varphi\in E$  the second term in~(\ref{eqn:aux18})
   describes oscillations of the system around its ``static reaction''~$\Lambda\varphi$.
Notice that for  $z\in \rho(\Ad)$ the second term is a bounded operator in~$E$.
Also of interest is the observation that the feedthrough operator~$\Lambda$ is
 independent of operators~$\Ad$ and $\Pi$
  describing oscillations, and therefore can be chosen
   to suit specific requirements of the given application.\footnote{
  See publications~\cite{BHMNW}, \cite{BMNW} as an example, where the authors modify the system's output
  by subtracting the static reaction,  thereby
  working with the system with the output defined
  as $(\gamma_1 - \Lambda\gamma_0)u^\varphi $
  and consequently with the null feedthrough operator.
 Here $\varphi \in E$ is the input and $u^\varphi\in H$ is the corresponding internal state.
}

%
%
%
Secondly,
 as described above, the operator of harmonic continuation~$\Pi: L^2(\Gamma) \to L^2(\Omega)$
 translates inputs into controls.
Its adjoint~$\Pi^*$ is called the \emph{observation operator }
   because according to~(\ref{eqn:aux18}) it maps internal states~$S_z \varphi = (I - z\Ad^{-1})^{-1}\Pi \varphi$ into the
   system's output, thereby making internal
    states available to the external observer.
 The equality $\Pi^* = \gamma_1 \Ad^{-1}$ is crucial for the representation~(\ref{eqn:aux18})  of M\nobreakdash-operator
 initially defined as ~$M(z) = \gamma_1 S_z$.
  For the model example of the Laplacian discussed above
   the identity~$\Pi^* = \gamma_1\Ad^{-1}$ is a
     consequence of Fubini's theorem and  properties of Green's function.
  To ensure validity of the representation~(\ref{eqn:aux18})
   within the general framework,
   the definition of abstract counterpart of~$\gamma_1$ given below
 explicitly involves operator~$\Pi^*$, see Definition~\ref{defn:Gn}.
%
%

%
Finally,
 from the theoretical point of view the system is considered a ``black box,'' with the transfer function being the
  only source of information about its internals available to the observer.
It follows that the linear system
  defined by equations~(\ref{eqn:aux3}) or (\ref{eqn:aux16}) with the
   internal states-outputs map~$\gamma_1$
    is completely described by the operators~$\Ad^{-1}$, $\Pi$, and $\Lambda$
     participating in the
      representation~(\ref{eqn:aux18}) of its transfer function.
In other words, the study of~(\ref{eqn:aux3}) from the systems theory perspective
 is equivalent to the study of the set~$\{\Ad^{-1}, \Pi, \Lambda\}$.
%


The transition from the system defined in terms of~$\{A, \gamma_0, \gamma_1\}$
 to the system defined by~$\{\Ad^{-1}, \Pi, \Lambda\}$ is known in the systems
  theory as \emph{reciprocal transform},
   see~\cite{Cur}, \cite{Staff} and \cite{Ryz2}.
 These two systems share the state and input-output spaces, their transfer functions coincide, but
  their defining operators are different.
 One advantage of the reciprocal transform is that it translates operators~$\{A, \gamma_0, \gamma_1\}$
  that are often difficult to describe in practical applications into the set of
    well defined and closed operators~$\{\Ad^{-1}, \Pi, \Lambda\}$, two of which are bounded.
For instance, the Laplacian~$A = -\Delta$ in the domain~$\Omega$ from the model example above, when
  defined in its ``natural domain,'' that is, the Sobolev space~$H^2(\Omega)$,
   is not a closed operator in $L^2(\Omega)$.
   At the same time the mappings~$\gamma_0$ and $\gamma_1$ are well
    defined on~$H^2(\Omega)$, although they are not closed on their domains  either.
The procedure of operator theoretic closure of~$A = -\Delta$ in the space $L^2(\Omega)$ results
 in the operator~$\overline A = \mathop{clos}(-\Delta)$
  with the domain that contains elements from $L^2(\Omega)\setminus H^2(\Omega)$.
Because the null set of a closed operator is always closed,
 $\Dom(\overline A)$ contains at least the $L^2$-closure of all harmonic functions
  continuous in~$\overline \Omega$.
 This set includes functions that need not to possess boundary values on $\Gamma$, so that
   the boundary mappings~$\gamma_0$, $\gamma_1$ can not be defined on all
    elements from $\Dom(\overline A)$.
Therefore the choice of suitable domain for~$A$ and subsequent expressions for boundary operators are not always obvious
 (except for the simplest cases, involving a boundary space of finite dimensionality as one example).
In contrast, operators of the reciprocal system~$\{\Ad^{-1},\Pi,\Lambda\}$ are all well defined and always
 closed.
They are the solution operator of the Dirichlet problem in the domain~$\Omega$, the operator
 of harmonic continuation from the boundary~$\Gamma$ into $\Omega$,
 and the classical Dirichlet-to-Neumann map for the Laplacian in~$\Omega$, respectively.


References to the reciprocal transform also help to clarify the relationship between the paper's
  framework and the mentioned earlier approach based on the notion of
   boundary triples.
The starting point for the latter is  the
    set~$\{\overline{A}, \gamma_0, \gamma_1\}$ (where $\overline A$ is the operator-theoretic closure of $A$)
   that gives rise to an abstract Green's formula,
    as opposed to the discussion below carried out on the basis of operators~$\{\Ad^{-1}, \Pi, \Lambda\}$
     that define the ``reciprocal'' system.
In order to circumvent the described above difficulties with the operator domains
 the earlier versions of boundary triples approach~\cite{DM1}, \cite{DM2}, \cite{DM3}
 severely limited its applicability
 by requesting the operator~$A$ to be closed, $\gamma_0$, $\gamma_1$ to be bounded in the graph norm of $A$,
 and the ranges of $\gamma_0$, $\gamma_1$ to coincide with the boundary space~$E$.
The last assumption is the most restrictive, as it automatically excludes
 from consideration unbounded M\nobreakdash-functions.
These limitations were removed only recently, see papers~\cite{BehrLan}, \cite{DHMS2},
 opening further possibilities of non trivial applications to the partial differential operators.
In contrast, the approach based on the set~$\{\Ad^{-1}, \Pi, \Lambda\}$
 offers a framework  free of these restrictions.
It not only allows one to work with closed and bounded operators,
 but also gives an option
 to selectively choose inputs from the boundary space~$E$, thus eliminating the concern
  of a suitable domain definition for boundary mappings
   and removing the assumption of closedness (and even closability) of~$A$.
It is also worth mentioning that when  operators~$\{\overline A, \gamma_0,\gamma_1\}$ form a ''boundary triplet,''
 all three of them are mutually interdependent.
Their domains must be suitably chosen and their definitions
  must fit together in order for the Green's formula to hold.
In the ``reciprocal'' approach, only two operators, $\Ad$ and~$\Pi$,
 are interdependent (the intersection of their ranges must be trivial),
  whereas the operator~$\Lambda$ (both its action and its domain) 
   can be selected arbitrarily.
Qualitatively speaking, one may say
 that the boundary triples method goes ``from the inside to the outside'' relating
 elements of the state space~$H$ to
  elements in the boundary space~$E$ by means of operators~$\gamma_0$ and $\gamma_1$,
   whereas the approach adopted in this paper
    goes in the opposite direction
     by introducing the control operator~$\Pi$ that maps elements from the boundary space
      into elements of the state space.
Operator~$\Lambda$ then, as a feedthrough operator acting on the boundary space,
 is an arbitrary parameter that does not have to be closed and even closable.


\paragraph{Applications}
Translation of classic boundary value problems and their solution procedures
 to the operator theoretic language suggests applicability of the obtained results 
 in various settings.
 As one example,
 it seems  rather natural to consider a more general
 type of  boundary conditions~(\ref{eqn:aux5})
 written as $(\alpha \gamma_1 + \beta \gamma_0) u= 0$
 with some linear operators $\alpha$, $\beta$  acting on $L^2(\Omega)$
  (or even bounded operator valued functions ~$\alpha(z)$, $\beta(z)$
   of the spectral parameter~$z\in \Complex$).
If $\beta = \chi^{}_E$ is the characteristic function
 of a non empty measurable set~$E\subset\Gamma$
 of positive Lebesgue  surface measure on~$\Gamma$
 and $\alpha = 1 - \chi^{}_E$, then
 the boundary condition above
 takes the form $(1 - \chi^{}_E)\partial_\nu u+ \chi^{}_{E} u |_\Gamma =0$.
It describes the so called mixed boundary value problem (Zaremba's problem)
 with the Dirichlet boundary
  condition  on $E$ and the Neumann condition
  on $\Gamma \setminus E$~(cf.~\cite{ML}).
%

%
%
The abstract operator theoretic technique
  elaborated in the paper can be successfully applied
     to the study of boundary value problems of classic and modern
      complex analysis.
In particular, it is possible to
  reformulate within the abstract framework
   classic problems of Poincar\'e,
    Hilbert, and Riemann for harmonic and analytic
     functions in bounded simply connected and
      sufficiently smooth domains of the complex plane, see~\cite{Ryz6}.
The generic boundary conditions in the form~$(\alpha \gamma_1 + \beta \gamma_0) u= 0$
 appear rather naturally in these cases.



One more example is based on the earlier study~\cite{Ryz2} and is discussed here at some length.
Using the above notation,
 it concerns the transmission type boundary condition
  imposed on solutions to the equation $(-\Delta - \zeta I)u = 0$
   inside and outside of~$\Omega$.
It is convenient
to rewrite this equation as $(A - zI)u = 0$ with $A = -\Delta + I$ and $z = \zeta + 1$
for reasons that will be clarified shortly.
Denote~$u^\pm_z$ its solutions in the domains~$\Omega^\pm$ where $\Omega^- = \Omega$ and $\Omega^+ = \Real^n \setminus \overline\Omega$.
Then
  the boundary condition $(\partial_\nu u^-_z)|_\Gamma - (\partial_\nu u^+_z)|_\Gamma = \varphi$ with $\varphi \in L^2(\Gamma)$
 defines a variant of transmission problem.
  Here~$(\partial_\nu u^\pm_z)|_\Gamma$ are boundary values on $\Gamma$ of the
  normal derivatives of functions~$u_z^\pm$ in the direction of
  outer normal to the domain~$\Omega$.
The solution to this problem is given by the single layer potential
\[
 (\mathscr S_z \varphi)(x)   := \int_\Gamma
  G(x,y, z) \varphi (y) d\sigma_y, \qquad x\in \Real^n
\]
where $G(\cdot, \cdot, z) $ is the standard Green's function of the differential operator~$(-\Delta +I  - zI)$ and
$d \sigma_y$ is the Euclidian surface measure on~$\Gamma$.
In order to include this problem into the paper's framework, define operators~$\gamma_0$ and $\gamma_1$
acting on linear combinations of smooth functions~$v^\pm \in L^2(\Omega^\pm)$ with the property $v^-|_\Gamma = v^+|_\Gamma = v|_\Gamma \in C(\Gamma)$
 as maps
 \[
 \gamma_0: v \mapsto  (\partial_\nu v^-)|_\Gamma - (\partial_\nu v^+)|_\Gamma, 
 \qquad 
  \gamma_1: v \mapsto  v|_\Gamma
 \]
where we put~$v := v^+ + v^- \in L^2(\Real^n)$.
Properties of single layer potentials are such that boundary values on~$\Gamma$
of the function~$\mathscr S_z \varphi$ taken from $\Omega^+$ and $\Omega^-$
coincide almost everywhere.
Moreover, the difference of boundary values of normal derivatives of~$\mathscr S_z \varphi$
from inside and outside of~$\Omega$ are  equal to $\varphi$ almost everywhere.
In other words, $\gamma_0 \mathscr S_z \varphi$ and $\gamma_1 \mathscr S_z \varphi$ are well defined and
$\gamma_0 \mathscr S_z \varphi = \varphi$.
Now it is only a matter of interpretation to treat this transmission problem
 as a spectral problem in the form~(\ref{eqn:aux3}).
The solution operator~$S_z$ obviously coincides with~$\varphi \mapsto \mathscr S_z \varphi$ and the choice of operator~$\gamma_1$
made above leads to the M\nobreakdash-function being the single layer potential restricted
to~$\Gamma$, that is, $M(z)\varphi = \mathscr S_z \varphi|_\Gamma$.
Corresponding expressions for $\Pi$ and $M(0) = \gamma_1 \Pi$  easily follow from their definitions.
More precisely, since $\Pi = \mathscr  S_z|_{z=0}$ we have $\Pi \varphi = \mathscr S_0 \varphi$ and
$M(0) \varphi = \mathscr S_0 \varphi |_\Gamma$.


The expression for~$\Ad = A|_{\Ker(\gamma_0)}$ deserves further discussion.
Since $A$ is initially defined on the domain of all functions $v \in L^2(\Real^n)$, smooth in $\Omega^\pm$ and continuous in $\Real^n$,
the condition~$\gamma_0 v =0 $ makes $\Ad$ equal to $-\Delta + I$ defined on the domain of standard  Laplacian~$-\Delta$
in $L^2(\Real^n)$ (after the conventional operator closure procedure).
This fact follows from the embedding theorems for Sobolev classes~$H^2$,
 according to which the function~$v = v^- + v^+ $, where  $v^\pm \in H^2(\Omega^\pm)$ belongs to  $H^2(\Real^n)$
  if
   $v^-|_\Gamma = v^+|_\Gamma$ and
  $(\partial_\nu v^-)|_\Gamma = (\partial_\nu v^+)|_\Gamma$ almost everywhere on $\Gamma$.
Also note that the addition of identity operator~$I$ to the Laplacian~$-\Delta$ ensures bounded invertibility of~$\Ad$.
Operator defined by~(\ref{eqn:aux5}) with $\beta =0$ (that  is, by the condition ~$\gamma_1 u = 0$)
is the orthogonal sum of two Dirichlet Laplacians acting in~$L^2(\Omega^-)\oplus L^2(\Omega^+)$.
 A more general transmission problem corresponding to the
 boundary condition~$\alpha (v|_\Gamma) + \beta[(\partial_\nu v^-)|_\Gamma - (\partial_\nu v^+)|_\Gamma ] = \varphi $
  with $\varphi \in L^2(\Gamma)$ and bounded operators $\alpha$, $\beta$ acting in $E = L^2(\Gamma)$
   is a particular case of problems investigated in the present paper.


It is also clear that the setting of transmission problem can be interpreted as a case of
 singular perturbations of quantum mechanics~\cite{AFHKL}, \cite{AGHH}, \cite{AlbevKur},
  where the ``free'' Laplacian defined initially
   in all space~$\Real^n$ is perturbed by the ``potential'' supported by the surface~$\Gamma$.
 Various boundary conditions in the form~$(\alpha\gamma_0 + \beta\gamma_1) u = 0$ with $\gamma_0$, $\gamma_1$ as above and
   suitable choice of linear operators $\alpha$, $\beta$ acting in $L^2(\Gamma)$ reflect the variety of possible ``parameterizations'' available in this model.
Another illustration of the point of view based on the theory of singular perturbations is given in the last section.


Naturally, the same considerations are applicable to more generic elliptic differential operators in place of the Laplacian, as long as the single layer potential
constructed by the Green's function of such operators possesses the same boundary properties as the conventional ``acoustic'' potential~$\mathscr S_z$, see~\cite{Agran}, \cite{Maz}, \cite{ML}.
 In particular, the Schr\"odinger operator~$-\Delta + q(x)$ in $L^2(\Real^n) $ with  sufficiently regular real valued function~$q(x)$ satisfies this condition.
 It is a remarkable fact that  when $n =3$, $q\in L^\infty(\Real^3)$ and $\Omega = \{x\in \Real^3\mid |x| < 1\}$
  the M\nobreakdash-function defined by the theory elaborated in the paper coincides with the
    Weyl-Titchmarsh function of
    the three-dimensional Schr\"odinger operator obtained in~\cite{AmP}
     by the multidimensional analogue of the classical nesting
      procedure of the Sturm-Liouville theory~\cite{Titch} (see~\cite{Ryz2} for the proof).
 Thus the single layer potential constructed by the Green's function of Schr\"odinger operator with the density
 supported by the unit sphere in~$\Real^3$ is a direct multidimensional equivalent
 of the celebrated Weyl-Titchmarsh~$m$-function.
%


A similarly developed theory for double layer potentials results in another type of transmission boundary conditions; the M\nobreakdash-function
 in this case coincides with the (unbounded) hypersingular integral operator acting in~$L^2(\Gamma)$.
The ``unperturbed'' operator $\Ad$ then is the ``free'' Laplacian acting in~$L^2(\Real^n)$,
whereas the operator defined by the condition ~$\gamma_1 u = 0$
is the orthogonal sum of two Neumann  Laplacians acting in~$L^2(\Omega^-)\oplus L^2(\Omega^+)$.
The interested reader is referred to the publication~\cite{Ryz2} for proofs and further details.

%

%
%
%
%



\section{Spectral Boundary Value Problem and its M-function} \label{ABST}
This section is concerned with a
 framework used in the study of spectral boundary value problems
   conducted in Sections~5 and~6.
A substantial part of the material covered here is an exposition of certain facts that
  can be found in the literature.
For the most general perspective,
 the  reader is referred to the works~\cite{DHMS1}, \cite{DHMS2}, \cite{DHMS3} and references therein
 carried out in a very generic setting of abstract boundary relations.
In fact, principal results communicated here can be derived from
 the exhaustive treatment of~\cite{DHMS2} as a particular case. 
Remark~\ref{rem:MM6}
 at the end of section outlines 
 a possible approach for such a derivation
 and  also clarifies existing
 relationships between~\cite{DHMS2}
 and the setting of present paper. 
 %
%
%
The main goal of this section is to give a
 concise account of all relevant facts in the form 
  convenient for the present study alongside with adequate proofs. 
%
%
%
Topics covered include
 the definition of
   spectral boundary value problem
  complemented by a discussion of properties of its solutions
  and the definition of corresponding
  M\nobreakdash-function.
An abstract analogue of the
 operator~$\gamma_1$ from Section~\ref{DNMAP}
  leading to the Green's formula and to the concept of
   weak solutions is elaborated in some depth.
The  study is conducted under the following assumption.


%
%

%
Let $H$, $E$ be two separable Hilbert spaces, $\Ad$ be a
 linear operator in
 $H$ defined on the dense domain~$\Dom(\Ad)$ in $H$
 and let $\Pi : E \to H$ be a bounded linear mapping.
\begin{assum}\label{assum:1}
 Suppose the following:
\begin{itemize}
  \item Operator $A_0$ is selfadjoint and boundedly invertible in
  $H$.
 \item
  Mapping~$\Pi$ possesses the
 left inverse~$\widetilde\Gamma_{\!0}$
 defined on~$\Ran(\Pi)$
 by $\widetilde\Gamma_{\!0} : \Pi \varphi \mapsto \varphi$, $\varphi \in E$.
 \item
  The intersection of~$\Dom(\Ad)$ and $\Ran (\Pi)$ is trivial,
  $\Dom(\Ad) \cap \Ran (\Pi) = \{0\}$.
\end{itemize}
\end{assum}
%


\begin{rem}\label{rem:MM1}
 As shown in~\cite{DHMS2}, conditions of Assumption~1
  can be substantially weakened.
  In particular, boundary mappings~$\Gd$ and $\Gn$ in the context of~\cite{DHMS2}
   are multivalued operators (linear relations) defined on the graph of operator~$A$
  that need not be single-valued, nor have a dense domain
   (compare with the definitions of ~$\Gd$, $\Gn$, and $A$ in our case below).
  In addition, bounded invertibility of ~$A_0$ is not required for validity of a number of statements found in this section.
%
%
\end{rem}


 Under Assumption~\ref{assum:1} neither of sets~$\Dom(\Ad) $ and $\Ran(\Pi)$
  coincides with
 the whole space~$H$.
  In follows that $\Ad$ is necessarily
 unbounded.
 Furthermore,
 existence of the left inverse of $\Pi$ implies $\Ker(\Pi) = \{0\}$.
The condition~$\Dom(\Ad) \cap \Ran (\Pi) = \{0\}$ is essential.
 It guarantees
 existence of (unbounded) projections from the direct sum
    $\Dom(\Ad) \dot{+} \Ran (\Pi)$ into the each component parallel to another.
In turn, it ensures correctness of definitions of operators~$A$ and $\Gd$
 in the next paragraph.
 Finally, note that for a non-invertible selfadjoint operator~$\Ad$
  with a
  real regular point~$c \in \rho(\Ad) \cap \Real$
  the invertibility condition can be easily satisfied
  by considering the operator~$\Ad - c I$ in place of $\Ad$.


Introduce two linear operators $A$ and $\Gd $ on the domain $\Dom(A) =
\Dom(\Gd)\subset H$ by
\begin{equation}
 \Dom(A) := \Dom(\Ad) \dot{+} \Ran (\Pi) = \{\Ad^{-1} f + \Pi \varphi \, |\, f \in H, \varphi \in E\}
\end{equation}
\begin{equation}\label{eqn:DefOfGamma0}
A : \Ad^{-1} f + \Pi \varphi \mapsto f,
 \quad
 \Gd : \Ad^{-1} f + \Pi \varphi \mapsto \varphi, \qquad
 f \in H, \varphi \in E
\end{equation}
Operators $A$ and $\Gd$ are extensions of $\Ad$ and
 $\widetilde\Gamma_{\!0}$ to $\Dom(A)$ defined to be the null mapping
 on the complementary
 subsets $\Ran(\Pi)$  and $\Dom(\Ad)$, respectively.
Observe that $\Ker(A) = \Ran(\Pi)$ and $\Ker(\Gd) = \Ran(\Ad^{-1}) $
($= \Dom(\Ad)$)
since  $\Ker(A|_{\Dom(\Ad)})$ and $\Ker(\Gd|_{\Ran(\Pi)})$ are
trivial by construction.
%


\begin{defn}
Spectral boundary problem associated with the pair~$\Ad$,
$\widetilde\Gamma_{\!0}$ satisfying Assumption~1 consists of the
system of linear equations for an unknown element~$u \in \Dom(A)$
\begin{equation} \label{def:BVP}
 \left\{
  \begin{aligned}
   (A -zI)u  & = f \\
    \Gd u  = \varphi
  \end{aligned}
 \right. \qquad \qquad f \in H, \varphi \in E
\end{equation}
where~$z\in \Complex$ is the spectral parameter.
\end{defn}


\begin{thm}\label{thm:BVPSolution}
For $z\in \rho(\Ad)$ and any $f\in H$, $\varphi \in E$ there exists
a unique solution~$u_z^{f,\varphi}$ to the problem~(\ref{def:BVP})
given by the formula
\begin{equation}\label{eqn:BVPSolution}
u_z^{f,\varphi}  = (\Ad - zI)^{-1}f + (I - z\Ad^{-1})^{-1}\Pi
\varphi
\end{equation}
Moreover, if for some $f \in H$ and $\varphi \in E$ the vector defined by the right hand side of~(\ref{eqn:BVPSolution}) is null, then $f =0$
and $\varphi =0$.
\end{thm}


\begin{proof}
We will show  that the first term in (\ref{eqn:BVPSolution}) is a
solution to the
 system~(\ref{def:BVP}) with $\varphi =0$, $f\neq 0$
 and the second one solves the
 system~(\ref{def:BVP}) for $f =0$, $\varphi \neq 0$.
To that end let us verify first
  that $(I - z\Ad^{-1})^{-1}\Pi \varphi$ belongs to $\Ker(A
 -zI)$.
 We have
 \begin{multline*}
 (A -zI)(I - z\Ad^{-1})^{-1}\Pi \varphi   =
 (A-zI)\left( I + z (\Ad -zI)^{-1}\right ) \Pi \varphi \\
  =
 \left ( A -zI + z(A- zI) (\Ad -zI)^{-1}\right ) \Pi \varphi
  =
  \left(A -zI + zI \right)\Pi \varphi = A\Pi \varphi  = 0
 \end{multline*}
since $\Ad \subset A$ and $\Ker(A) = \Ran (\Pi)$.
Therefore
\[
(A-zI)u_z^{f,\varphi} = (A-zI)(\Ad - zI)^{-1} f = f
\]
For the second equation~(\ref{def:BVP}) and $u_z^{f,\varphi}$ as
in~(\ref{eqn:BVPSolution}) ,
 \[
 \Gd u_z^{f,\varphi} = \Gd (I -z \Ad^{-1})^{-1}\Pi \varphi =
  \Gd (I + z (\Ad -zI)^{-1})\Pi \varphi = \Gd \Pi \varphi = \varphi
 \]
because~$\Ker(\Gd) = \Dom(\Ad) = \Ran((\Ad -z I)^{-1})$.
Both equations~(\ref{def:BVP}) are therefore satisfied.


Uniqueness of the solution~(\ref{eqn:BVPSolution}) is a direct consequence of
assumption~$z\in \rho(\Ad)$.
For $z =0 $ the implication $u_0^{f,\varphi} = 0 \Rightarrow f =0, \varphi =0$
trivially holds due to uniqueness of the decomposition $u_0^{f,\varphi} = \Ad^{-1} f + \Pi \varphi$ into the sum of two terms
 from disjoint sets and equalities~$\Ker(\Ad^{-1}) = \{0\}$, $\Ker(\Pi) = \{0\}$.
 For $z\in\rho(\Ad)$ with the help of
 identity $(I - z \Ad^{-1})^{-1} = I + z (\Ad -z I)^{-1}$
  the representation~(\ref{eqn:BVPSolution})
  can be rewritten as
\[
 u_z^{f,\varphi} = (\Ad - z I)^{-1} (f + z \Pi \varphi ) + \Pi \varphi
\]
The first summand here belongs to $\Dom(\Ad)$ and the second to
$\Ran (\Pi)$.
Since the intersection of these two sets is trivial,
 the equality $u_z^{f,\varphi}  =0 $ implies $\Pi\varphi = 0$ and
thus $\varphi = 0$.
Then $ (\Ad - z I)^{-1} f = 0 $ and therefore $f = 0 $.
\qed
\end{proof}


\begin{defn}
Assuming $z\in\rho(\Ad)$ denote $R_z = (\Ad -zI)^{-1}$ the resolvent
of~$\Ad$ and introduce the \textit{solution operator}~$S_z : E \to
E$
\[
 S_z : \varphi \mapsto (I - z\Ad^{-1})^{-1} \Pi \varphi = (I + zR_z)\Pi \varphi, \qquad
 \varphi \in E, \; z \in \rho(\Ad)
\]
\end{defn}

\begin{rem}\label{rem:MM4}
 An alternative name for the solution operator commonly accepted in the theory of linear symmetric operators and
  relations is \emph{$\gamma$-field}, see~\cite{DHMS1}, \cite{DHMS2}, \cite{DHMS3} and references therein.
   The present paper  follows the terminology inherited from the theory of boundary value problems~\cite{Grubb3}
   in order to stress out
     the role mapping~$S_z$ plays in the considerations below.
\end{rem}


\begin{rem}\label{rem:aux1}
Important properties of the solution operator follow from its definition and the resolvent identity
(see~\cite{DHMS2}, Proposition~4.11 for the general case).
Suppose   $z \in\rho(\Ad)$.
Then   $\Gd S_z  = I $ and
  $\Ran(S_z) = \Ker(A -zI)$.
Moreover,
\begin{equation}\label{eqn:RzMinusRzeta}
  S_z - S_\zeta = (z - \zeta)R_z S_\zeta,\qquad
  z,\zeta\in
 \rho(\Ad)
\end{equation}
\end{rem}
%


\begin{proof}
The first claim follows  from Theorem~\ref{thm:BVPSolution}.
The same theorem shows that the range of~$S_z$ is included into
$\Ker (A -zI)$.
To show that $\Ran(S_z) = \Ker(A -zI)$  assume  $u = \Ad^{-1} f +
\Pi \varphi$ with $f\in H$, $\varphi \in E$ is such that $ u \in
\Ker(A -zI)$.
 Then
\[
 0 = (A - zI) u = (A -zI)(\Ad^{-1} f + \Pi \varphi )
 = (I - z \Ad^{-1}) f - z \Pi \varphi
\]
so that $f = z (I - z\Ad^{-1})^{-1}\Pi\varphi$.
Substitution into $u = \Ad^{-1} f + \Pi \varphi$ gives
\[
u = \Ad^{-1} f + \Pi \varphi = \left[z\Ad^{-1} (I - z\Ad^{-1})^{-1}
+ I \right] \Pi \varphi
 = (I - z\Ad^{-1})^{-1}\Pi \varphi = S_z \varphi
\]
The last statement is easily verified by the direct calculation
based on the resolvent identity
\[
 \begin{aligned}
 (I -   z\Ad^{-1})^{-1}   & -
 (I - \zeta\Ad^{-1})^{-1}
 =
  z (\Ad - zI)^{-1} - \zeta (\Ad - \zeta I )^{-1}
  \\
  & = (\Ad - zI)^{-1}\left( z I - \zeta (I - z \Ad^{-1})(I - \zeta \Ad^{-1})^{-1} \right)
  \\
  & = (\Ad - zI)^{-1}\left( z (I - \zeta \Ad^{-1}) - \zeta (I - z \Ad^{-1}) \right) (I - \zeta \Ad^{-1})^{-1}
  \\
  & = (z - \zeta)(\Ad - zI)^{-1} (I - \zeta \Ad^{-1})^{-1}
 \end{aligned}
\]
Multiplication by $\Pi$ from the right concludes the proof.
\qed
\end{proof}


Now an analogue of the ``second boundary
 operator''~$\gamma_1$ described in Section~1 can be introduced.
%


\begin{defn}\label{defn:Gn}
 Let $\Lambda$ be a linear operator in $E$ with the
 domain~$\Dom(\Lambda) \subset E$.
 Define the linear mapping~$\Gn$
 on the subset~$\mathscr D := \Dom(\Ad) \dot{+} \Pi \Dom(\Lambda)$
 by
\begin{equation}\label{eqn:DefOfGamma1}
 \Gn : \Ad^{-1} f + \Pi \varphi \mapsto \Pi^* f + \Lambda \varphi,
  \qquad f \in H, \varphi \in \Dom(\Lambda)
\end{equation}
\end{defn}


Note that according to this definition $\Lambda = \Gn \Pi$ and $\Pi
= (\Gn \Ad^{-1})^*$.
 In particular, for the solution
 operator~$S_z = (I - z\Ad^{-1})^{-1}\Pi = \Ad(\Ad - zI)^{-1}\Pi $ we obtain
\begin{equation}\label{eqn:GnADz-1}
 (S_{\bar z})^*
  = \Gn (\Ad -zI)^{-1} = \Gn R_z, \quad z\in \rho(\Ad)
\end{equation}
%


\begin{assum}\label{assum:2}
 Operator $\Lambda = \Gn \Pi$ is selfadjoint (and thereby densely defined).
\end{assum}

\begin{rem}
 In the sequel it is always assumed that the set $\{\Ad^{-1}, \Pi,\Lambda \}$
 satisfies both Assumptions~\ref{assum:1} and~\ref{assum:2}.
\end{rem}


\begin{thm}[Green's Formula]\label{thm:GreenFormula}
 \[
  (A u, v)_H - (u, Av)_H = (\Gn u, \Gd v)_E - (\Gd u, \Gn v)_E,
  \quad u,v \in \mathscr D
 \]
\end{thm}
%


%
\begin{proof}
 Let $u = \Ad^{-1} f + \Pi \varphi$, $v = \Ad^{-1} g + \Pi \psi$
 with $f,g \in H$, $\varphi, \psi \in \Dom(\Lambda)$.
 We have $A u = f$, $A v = g$, and
 due to selfadjointness
 of $\Ad^{-1}$ and $\Lambda$,
 \begin{multline*}
 (A u, v)_H - (u, A v)_H   = (f, \Ad^{-1} g + \Pi \psi) - (\Ad^{-1} f + \Pi \varphi, g)
  = (f, \Pi \psi) - (\Pi \varphi, g) = \\
    (\Pi ^* f,  \psi) - (\varphi, \Pi^*g)
  = (\Pi ^* f + \Lambda \varphi,  \psi) - (\varphi, \Pi^*g + \Lambda \psi)
  = (\Gn u, \Gd v) - (\Gd \varphi, \Gn v)
 \end{multline*}
since both Assumptions~1 and~2 are valid.
\qed
\end{proof}

Introduction of the second boundary operator~$\Gn$
 and Theorem~\ref{thm:GreenFormula} lead to the
 concept of weak solutions to the problem~(\ref{def:BVP})
 defined as solutions to a certain ``variational'' problem.
%
%


%
\begin{defn}\label{def:WeakSolution}
 The
 weak solution of the problem~(\ref{def:BVP})
 is an element $w_z^{f,\varphi} \in H$
 satisfying
\begin{equation}\label{eqn:WeakEquation}
  (w_z^{f,\varphi}, (\Ad -\bar z I)v) = (f,v) + (\varphi , \Gn v)
  \qquad \text{ for any} \quad v \in \Dom(\Ad)
\end{equation}
\end{defn}
%


Let us verify that this definition is consistent with the
solvability statement of Theorem~\ref{thm:BVPSolution}.
In other words, we need to show  that
 for $z\in \rho(\Ad)$ the vector~$u_z^{f,\varphi}$
 from (\ref{eqn:BVPSolution})
 solves the variational problem~(\ref{eqn:WeakEquation}).
Indeed,  for $u_z^{f,\varphi} = R_z f + S_z\varphi$ and any $v\in
\Dom(\Ad)$ we have
\[
\begin{aligned}
 (u_z^{f,\varphi}, (\Ad -\bar zI)v)  & = ( R_z f, (\Ad -\bar zI)v) + (S_z \varphi, (\Ad -\bar zI)v)
  \\
 &  = (f,v) + (\varphi , (S_{z})^{*}(\Ad -\bar zI)v
  = (f,v) + (\varphi, \Gn  v)
\end{aligned}
\]
according to~(\ref{eqn:GnADz-1}) and the claim is proved.
%

\begin{rem}\label{rem:WeakSolutions}
The notion of weak solution suggests that the applicability of
representation~(\ref{eqn:BVPSolution})
 is wider than that described in Theorem~\ref{thm:BVPSolution}.
 Firstly, rewrite the right hand side of~(\ref{eqn:WeakEquation}) as
 \begin{equation}\label{eqn:aux2}
 (f,v)_H + (\varphi, \Gn  v)_E =
 (\Ad^{-1}f, \Ad v) + (\varphi, \Pi^* \Ad v)
 =
 (\Ad^{-1}f  + \Pi \varphi, \Ad v)
 \end{equation}
Recall now that $\Ran(\Ad) =H$.
Therefore
 the concept of weak solutions can be extended to the case when
 $f$ and $\varphi$ are chosen from spaces wider than~$H$  and $E$
 as long as the
 sum~$\Ad^{-1}f  + \Pi \varphi$
 belongs to $H$.
 As an illustration consider a simple
 example when $f$ and $\varphi$
 are such that both summands
 on the left side of~(\ref{eqn:aux2})
 are finite.
 Let $H_- \supset H$ and $E_- \supset E$ be Hilbert spaces
 obtained by completion of $H$ and $E$ with respect
 to norms $\|f\|_- = \|\Ad ^{-1}f\|_H$ and $\|\varphi\|_- =
 \|\Pi \varphi\|_H$, where   $f\in H$, $\varphi \in E$,
 correspondingly.
Since both $\Ker(\Ad^{-1})$ and $\Ker(\Pi)$ are trivial, these norms
are
 non-degenerate.
 For each $v\in \Dom(\Ad)$ the usual estimates hold
 \[
 \begin{aligned}
 |(f,v)| & \leq \|\Ad^{-1} f\|\cdot \|\Ad v\| = \|f\|_-\cdot \|\Ad v\|
 \\
  |(\varphi ,\Gn v)| & = |( \varphi, \Gn \Ad^{-1} \Ad v)| = |( \Pi\varphi,  \Ad v)
  | \leq \|\Pi\varphi\|\cdot \|\Ad v\| = \|\varphi\|_-\cdot \|\Ad
  v\|
 \end{aligned}
 \]
Thus the right hand side of~(\ref{eqn:aux2}) is finite for any~$v \in
\Dom(\Ad)$ so that $\Ad^{-1} f + \Pi \varphi \in H$ as long as $f\in
 H_-$ and $\varphi\in E_-$.
It follows that
 the vector~$u_z^{f,\varphi} = R_z f + S_z\varphi$
 defined  for  $z\in \rho(\Ad)$
 by the formula~(\ref{eqn:BVPSolution})
 is the weak solution of (\ref{def:BVP}) with $f\in
 H_-$, $\varphi\in E_-$.
\end{rem}

%
%
%
%
%
%
%
%

%
Introduce the notion of M-function (M-operator) as follows.


\begin{defn}\label{def:Mz}
Operator-valued function~$M(z)$ defined on the
 do\-main~$\Dom(\Lambda)$ for $z\in\rho(\Ad)$ by the formula
 \[
  M(z) \varphi = \Gn S_z\varphi = \Gn (I - z\Ad^{-1})^{-1}\Pi
  \varphi
 \]
is called the M-function of the problem~(\ref{def:BVP}). 
\end{defn}


\begin{thm}\label{thm:PropertiesOfMz}
 \begin{enumerate}
 \item
The representation is valid
\begin{equation}\label{eqn:MFunction}
 M(z) = \Lambda + z \Pi^*(I - z\Ad^{-1})^{-1}\Pi, \qquad z\in \rho(\Ad)
\end{equation}
\item For each $\varphi \in \Dom(\Lambda)$ the vector function $M(z)\varphi$, $z \in \rho(\Ad) $ with values in $E$ 
is analytic for.
 \item
 For $z, \zeta \in \rho(\Ad)$ the operator $M(z) - M(\zeta)$ is
 bounded and
 \[
 M(z) - M(\zeta) = (z - \zeta)(S_{\bar z})^*S_\zeta
 \]
  In particular,
 $
 \I M(z) = (\I z) (S_{\bar z})^*S_z
 $
 and
  $(M(z))^* = M(\bar z)$
  where $\I M(\cdot)$ denotes the imaginary part of operator $M(\cdot)$.
 \item
  For $u_z \in \Ker(A -zI)\cap \mathscr D = \Ker (A -zI)\cap \{\;\Dom(\Ad)\dot{+} \Pi \Dom(\Lambda)\}$
 the formula holds
 \begin{equation}\label{eqn:MG0=G1}
  M(z)\Gd u_z = \Gn u_z
 \end{equation}
 \end{enumerate}
\end{thm}


\begin{proof}
(1) The claim  follows from the identities~$\Lambda = \Gn \Pi$,
$\Pi^* = \Gn \Ad^{-1}$, the elementary computation
\[
 (I - z\Ad^{-1})^{-1} =  I + z (\Ad -z I)^{-1} =
 I + z \Ad^{-1} (I - z\Ad^{-1})^{-1}, \quad
 z\in\rho(\Ad)
\]
and the definition $M(z) = \Gn (I - z\Ad^{-1})^{-1}\Pi$.
%


(2)  As the term~$z\Pi^*(I -z \Ad^{-1})^{-1}\Pi$ is a bo\-un\-ded
analytic operator-function of $z\in\rho(\Ad)$ the statement is a
consequence of the representation obtained in~(1).


(3) We have
\[
 \begin{aligned}
 M(z) - M(\zeta) & = \Pi^* \left[ z(I - z\Ad^{-1})^{-1} - \zeta(I - \zeta\Ad^{-1})^{-1}\right] \Pi
 \\
 & =
 \Pi^* (I - z\Ad^{-1})^{-1} \left[ z(I - \zeta\Ad^{-1}) - \zeta(I - z\Ad^{-1})\right] (I - \zeta\Ad^{-1})^{-1}\Pi
 \\
 &= (z - \zeta) \Pi^* (I - z\Ad^{-1})^{-1} (I - \zeta\Ad^{-1})^{-1}\Pi
= (z - \zeta) \left(S_{\bar z}\right)^* S_\zeta.
 \end{aligned}
\]
 The equality~$(M(z))^* = M(\bar z)$ is valid
due to  selfadjointness of $\Lambda$.


(4) Any vector~$u_z\in\Ker(A -zI)$ is uniquely represented in the
form $u_z = S_z\Gd u_z$.
In the case~$u_z \in \mathscr D$ either side belongs to $\Dom(\Gn)$.
Therefore, $\Gn u_z = \Gn S_z\Gd u_z = M(z)\Gd u_z$.
\qed
\end{proof}

\begin{rem}\label{rem:MM6}
Results of~\cite{DHMS2}
 suggest an alternative approach to build the
  framework described in this section.
 As an illustration of this possibility, and
 in order to explain  relationships between~\cite{DHMS2}
  and the present paper,
   let us derive the representation~(\ref{eqn:MFunction})
     for Weyl function~$M(z)$ within the scope of~\cite{DHMS2}.
 The key component here is the Example~6.6 of~\cite{DHMS2}.
 Using notations of this Example,
  substitution of $D = \Ad^{-1}$, $B = \Pi$, and $E = - \Lambda$
   yields the following form of boundary relation~$\Gamma : H\oplus H \to E \oplus E$
 \[
 \Gamma =  \left\{ \binom{f}{\Ad^{-1} f + \Pi \varphi}, \binom{\varphi}{-\Lambda\varphi - \Pi^* f} \right\}, \qquad
   f \in H, \, \varphi \in \Dom(\Lambda)
 \]
Formula~(3.6) of~\cite{DHMS2} splits~$\Gamma$ into two
boundary mappings, $\widehat{\Gd}$  and $\widehat{\Gn}$
\[
 \widehat{\Gd} = \left\{ \binom{f}{\Ad^{-1} f + \Pi \varphi}, \binom{\varphi}{0} \right\}, \qquad
 \widehat{\Gn} = \left\{ \binom{f}{\Ad^{-1} f + \Pi \varphi}, \binom{0}{ - \Lambda \varphi - \Pi^* f} \right\}
\]
where~$f\in H$, $\varphi \in \Dom(\Lambda)$.
Note that the mapping $\widehat {\Gd}$ can be extended to the subset $\{f, \Ad^{-1} f + \Pi \varphi\}$ with ~$f \in H$, $\varphi \in E$.
Comparison to expressions~(\ref{eqn:DefOfGamma0}) and (\ref{eqn:DefOfGamma1})
 for operators~$\Gd$ and~$\Gn$
  clarifies relationships between $\{\widehat{\Gd}, \widehat{\Gn}\}$ and $\{\Gd, \Gn\}$.
  More precisely, for~$f\in H$, $\varphi \in \Dom(\Lambda)$
\begin{align*}
 \widehat{\Gd} & = \left\{ \binom{f}{\Ad^{-1} f + \Pi \varphi}, \binom{\Gd \left(\Ad^{-1} f + \Pi \varphi \right)}{0} \right\}, 
\\
 \widehat{\Gn} & = \left\{ \binom{f}{\Ad^{-1} f + \Pi \varphi}, \binom{0}{ - \Gn \left(\Ad^{-1} f + \Pi \varphi \right)} \right\}
\end{align*}



Weyl family~$\widehat M(z)$ corresponding to~$\Gamma$ is the relation
\[
 \widehat M(\lambda) = \left\{
  \widehat \varphi \in E \oplus E \mid \{ {\widehat f}_\lambda, \widehat \varphi \} \in \Gamma
   \text {    for some    } {\widehat f }_\lambda = \{f, \lambda f\} \in H \oplus H
 \right\}
\]
(see Definition~3.3 of~\cite{DHMS2}).
For any element~${\widehat f }_\lambda = \{f, \lambda f\} \in H \oplus H$ the
 condition~$\{ {\widehat f}_\lambda, \widehat \varphi \} \in \Gamma $ implies ${\widehat f}_\lambda \in \Dom(\Gamma)$,
  which leads to the
 equation $  \lambda f = \Ad^{-1} f + \Pi \varphi$
  for vectors $f$ and~$\varphi$.
It follows that $f = (\lambda I - \Ad^{-1})^{-1} \Pi \varphi$ at least for $\I ( \lambda) \neq 0 $.
If this equality holds, then the relation~$\Gamma$ takes the form
\[
 \Gamma =  \left\{ \binom{f}{\lambda f}, \binom{\varphi} {-\Lambda\varphi - \Pi^* (\lambda I - \Ad^{-1})^{-1} \Pi \varphi} \right\},
 \qquad   f \in H, \, \varphi \in \Dom(\Lambda)
\]
and therefore the Weyl family is the relation defined  for~$\varphi \in \Dom(\Lambda)$ as 
\[
\widehat M(\lambda) = \left\{\varphi, -( \Lambda + \Pi^* (\lambda I - \Ad^{-1})^{-1} \Pi) \varphi \right\}
\]
Additionally,  decomposition of~$\Gamma$ into two boundary mappings~$\widehat {\Gd}$ and $\widehat {\Gn}$ yields
 \[
 \widehat{\Gd} = \left\{ \binom{f}{\lambda  f}, \binom{ \varphi}{0} \right\}, \qquad
 \widehat{\Gn} = \left\{ \binom{f}{\lambda f }, \binom{0}{ - ( \Lambda  + \Pi^*(\lambda I - \Ad^{-1})^{-1}\Pi )\varphi } \right\}
 \]
 and therefore $\widehat{\Gn} \widehat f_\lambda = \widehat M(\lambda)\widehat{\Gd} \widehat f_\lambda $
  for any~$\widehat f_\lambda = \{f, \lambda f\} \in \Dom(\Gamma)$
  (cf. (3.7) of~\cite{DHMS2} and (\ref{eqn:MG0=G1}) above).


 Finally, the relation~$\widehat M(\lambda)$ is the graph of a linear operator in $E$ (also denoted $\widehat M(\lambda)$) with the domain~$\Dom(\Lambda)$
 and   $M(z) = - \widehat M (1/z)$, $\I (z) \neq 0$ where $M(z)$ is the M-function~(\ref{eqn:MFunction}) of  boundary value problem~(\ref{def:BVP}).
 \end{rem}

%
%
%
%
%
%
%
%
%

\section{Boundary Conditions}\label{BOUNDC}

This section explores other types of boundary value problems
 for the operator~$A$ and boundary mappings~$\Gd$, $\Gn$
 introduced in Section~\ref{ABST}.
 The problems under consideration are defined in terms of
 certain linear ``boundary conditions.''
More precisely, given two linear operators $\beta_0$, $\beta_1$
 acting in the space~$E$
 we are formally looking for solutions to the equation~$(A -zI)u = f$
  satisfying condition
 $(\beta_0\Gd +\beta_1\Gn) u = \varphi$
 where $f\in H$, $\varphi \in E$, and $z\in \Complex$.
The exact meaning of this problem statement and
 the solvability theorem are the main results of this section.
Definitions and some
  properties of associated $M$\nobreakdash-functions are also briefly reviewed.


Everywhere below $\beta_0$, $\beta_1$ are two linear operators in
$E$ such that $\beta_0$ is defined on the
domain~$\Dom(\beta_0)\supset \Dom(\Lambda)$ and $\beta_1$ is defined
everywhere on~$E$ and bounded.
Consider the following spectral boundary  value problem for $w\in H$
associated
 with the set~$\{\Ad^{-1}, \Pi, \Lambda\}$ and the
  pair~$(\beta_0,\beta_1)$
\begin{equation}\label{eqn:BPVwithBeta}
 \left\{
  \begin{gathered}
    (A -zI)w  = f \\
    (\beta_0\Gd +\beta_1\Gn) w  = \varphi
  \end{gathered}
  \qquad f \in H, \varphi \in E
 \right.
\end{equation}
where $z\in\Complex$ plays the role of a spectral parameter.


The first goal in the study of~(\ref{eqn:BPVwithBeta}) is
clarification of the equality
 $(\beta_0\Gd +\beta_1\Gn) w = \varphi$.
Having this objective in mind,  observe  that the sum~$\beta_0\Gd
+\beta_1\Gn$ is defined at least
 on~$S_z\Dom(\Lambda)$ for $z\in \rho(\Ad)$ and
\begin{equation}\label{eqn:beta0G0beta1G1_1}
 (\beta_0\Gd +\beta_1\Gn)S_z\varphi =  \left(\beta_0  + \beta_1
 M(z)\right)\varphi, \qquad \varphi \in \Dom(\Lambda)
 \end{equation}
according to the properties of~$S_z$ and definition of $M(z)$.
Rewrite the right hand side
 using
 the representation $M(z) = \Lambda + z \Pi^*(I - z\Ad^{-1})^{-1}\Pi$
 in the form
\begin{equation}\label{eqn:beta0G0beta1G1_2}
 (\beta_0\Gd +\beta_1\Gn)S_z\varphi = ( \beta_0  + \beta_1
 \Lambda) \varphi + z \Pi^* (I - z\Ad^{-1})^{-1}\Pi \varphi,
\qquad \varphi \in \Dom(\Lambda)
\end{equation}
The second term on the right  is bounded
for $z\in \rho(\Ad)$,
 thus the mapping properties of the sum~$\beta_0\Gd +\beta_1\Gn$
 as an operator from~$H$ into $E$ are fully determined
 by the map~$\beta_0 + \beta_1 \Lambda$.
The following closability condition is assumed
 to be always satisfied.


\begin{assum}\label{assum:3}
 The operator $\beta_0  + \beta_1 \Lambda$
 defined on $\Dom(\Lambda)$
 is closable in $E$.
 Let
 $\mathscr B = \overline{\beta_0  + \beta_1 \Lambda}$
 be  its closure.
\end{assum}

%
%
%
%
%
%

%
%
\begin{rem}\label{rem:B0B01MClosed}
 It follows from
 (\ref{eqn:beta0G0beta1G1_1}) and (\ref{eqn:beta0G0beta1G1_2})
 that under this assumption
 all  operators~$\beta_0 + \beta_1 M(z)$  are also closable
 for $z\in
 \rho(\Ad)$ and
  the domain of their closures coincides with $\Dom(\mathscr B)$.
Equality~(\ref{eqn:beta0G0beta1G1_2})
 therefore can be extended to the set $\varphi \in \Dom(\mathscr
 B)$.
 However,
 the operator sum~$\beta_0 \Gd + \beta_1 \Gn$
 needs not be
 closed on the linear set~$\{S_z\varphi \mid \varphi \in  \Dom(\mathscr B)\}$
 and in general cannot be treated  as a sum of two separate operators,
 $\beta_0 \Gd$ and  $\beta_1 \Gn$.
\end{rem}

\begin{defn}\label{defn:HB}
Let~$\mathscr H_{\mathscr B}$ be the linear set of elements
\[
\mathscr H_\mathscr B = \left \{
  \Ad^{-1} f + \Pi \varphi  \mid f\in H, \varphi\in \Dom(\mathscr
 B)\right\}
 \]
\end{defn}

Notice that since $\Dom(\Lambda)\subseteq \Dom(\mathscr B)\subseteq
E$, the inclusions $
 \mathscr D \subseteq \mathscr H_{\mathscr B} \subseteq \Dom(A) $
 hold, where  $\mathscr D =\{ \Ad^{-1} f + \Pi \varphi \mid f \in
 H, \varphi \in \Dom(\Lambda)\}$,
 as defined in Section~\ref{ABST}.


The set~$\mathscr H_\mathscr B$ can be turned into a (closed)
Hilbert space by
 introducing a certain non-degenerate metric.
Then the map $\beta_0 \Gd + \beta_1 \Gn$
 is bounded as an operator from~$\mathscr H_\mathscr B$
 into $E$.
 More precise result is given by the following Lemma.
%


\begin{lem}\label{lem:B0G0B1G1Boundedness}
 The set $\mathscr H_{\mathscr B}$
  is a Hilbert space
  with the norm
\[
\|u\|_{{\mathscr B} } =
 \left( \|f\|_H^2 +
       \|\varphi\|_E^2 +
       \|\mathscr B\varphi\|^2\right)^{1/2}.
\]
The operator~$\beta_0 \Gd + \beta_1 \Gn : \mathscr H_\mathscr B \to E$
  is bounded.
\end{lem}
%
%
%


\begin{proof}
The proof is based on the density of $\Dom(\Lambda)$ in the domain
 $\Dom(\mathscr B)$ equipped with the graph norm of
 operator~$\mathscr B$,
 which in turn implies density of $\mathscr D$
 in~$\mathscr H_\mathscr B$ in the norm~$\|\cdot\|_\mathscr B$.
%


Let $\{u_n\}_{n =1}^\infty\subset \mathscr D$ be a Cauchy sequence
in the norm of~$\mathscr H_\mathscr B$, that is~$\|u_n -
u_m\|_\mathscr B \to 0$ as $n, m \to \infty$.
Each vector $u_n$ is represented as the sum~$u_n = \Ad^{-1} f_n +
\Pi \varphi_n$ with uniquely defined $f_n \in H$, $\varphi_n \in
\Dom(\Lambda)$.
We have
\[
\|u_n - u_m\|^2_\mathscr B  = \|f_n  - f_m\|^2
 + \|\varphi_n - \varphi_m\|^2 +  \|\mathscr B(\varphi_n -
 \varphi_m)\|^2 \to 0 \text { as } n,m \to \infty
\]
The first summand here tends to zero, and therefore $f_n \to f_0 \in
H$ for some $f_0 \in H$   as $n\to\infty$.
The sum of second and third terms is the norm of $\varphi_ n -
\varphi_m$ in the
 graph norm of $\mathscr B$.
 Because operator~$\mathscr B$ defined on $\Dom(\Lambda)$ is closable,
 there exists a vector~$\varphi_0\in\Dom(\mathscr B)$
  such that $\varphi_n \to \varphi_0$ as $n\to\infty$.
The limit of the sequence~$\{u_n\}_{n = 1}^\infty$
 therefore is represented in the form~$\Ad^{-1} f_0 + \Pi \varphi_0$
  where $f_0 \in H$  and $ \varphi_0 \in \Dom(\mathscr B)$.
 Hence $ \mathscr H_\mathscr B $ is closed
 in the norm~$\|\cdot\|_{{\mathscr B}}$.


The second statement follows directly  from the norm estimate
 for elements of  $\mathscr D$.
When $f \in H$ and $\varphi \in \Dom(\Lambda)$, the sum $ u =
\Ad^{-1}f + \Pi\varphi$ belongs to the set~$\Dom(\Gd)\cap \Dom(\Gn)$
and
 \begin{align*}
  (\beta_0 \Gd + \beta_ 1\Gn) u & = \beta_1 \Gn \Ad^{-1} f
  + ( \beta_0\Gd + \beta_1 \Gn) \Pi \varphi 
\\
   & = \beta_1 \Pi^* f +
( \beta_0 +  \beta_1\Lambda) \varphi 
  = \beta_1 \Pi^* f + \mathscr B \varphi
\end{align*}
Because operator~$ \beta_1\Pi^*$ is bounded, the following estimates hold
\[
 \|(\beta_0 \Gd + \beta_ 1\Gn) u\| \leq C \|u
 \|_\mathscr B, \qquad  u= \Ad^{-1} f + \Pi \varphi, \; f \in H,\;
 \varphi \in \Dom(\Lambda).
\]
 The set~$\{\Ad^{-1} f + \Pi \varphi \mid f \in H, \varphi\in \Dom(\Lambda) \}$
 is dense in~$\mathscr H_\mathscr B$;
hence
 the operator $\beta_0
 \Gd + \beta_ 1\Gn$
 is bounded as a mapping from $\mathscr H_\mathscr B$ into $E$.
\qed
\end{proof}


\begin{rem}\label{rem:Symbol}
  The symbol~$\beta_0 \Gd + \beta_ 1\Gn$ will be used
   for the extension of operator of
   Lemma~\ref{lem:B0G0B1G1Boundedness}
   to the space $\mathscr H_\mathscr B$, although
   two terms in the sum~$(\beta_0 \Gd + \beta_ 1\Gn)u $
   need not exist separately
   for an arbitrary~$u\in \mathscr H_\mathscr B$.
\end{rem}


%
%
Taking Lemma~\ref{lem:B0G0B1G1Boundedness} into consideration, we
shall look for solutions of the problem~(\ref{eqn:BPVwithBeta}) that
belong to $\mathscr H_\mathscr B$.
%
%


\begin{thm}\label{thm:BVPwithBetSolution}
Suppose $z\in\rho(\Ad)$ is such that the closed
 operator $\overline{\beta_0 + \beta_1 M(z)}$ defined on~$\Dom(\mathscr B)$ is
 boundedly invertible in the space~$E$.
Then the problem~(\ref{eqn:BPVwithBeta}) is uniquely solvable
 and the solution~$w_z^{f,\varphi} \in \mathscr H_\mathscr B$
 is given by the formula
\begin{equation}\label{eqn:w_z}
w_z^{f,\varphi} = (\Ad - zI)^{-1} f + (I -
z\Ad^{-1})^{-1}\Pi\Psi_z^{f,\varphi}
\end{equation}
where $\Psi_z^{f,\varphi}$ is a vector from $\Dom(\mathscr B)$
\begin{equation}\label{eqn:Psi_z}
 \Psi_z^{f,\varphi} = (\overline{\beta_0 + \beta_1 M(z)})^{-1}
 (\varphi - \beta_1\Pi^*(I - z\Ad^{-1})^{-1}f )
\end{equation}
\end{thm}

%
%
%
%
%
%

\begin{rem}\label{rem:WeakSolForBetaBVP}
 According to this Theorem
 the problem~(\ref{eqn:BPVwithBeta})
is reduced to the problem~(\ref{def:BVP})
with $\varphi$ replaced by the vector~$\Psi_z^{f,\varphi}$
 defined in~(\ref{eqn:Psi_z}).
 This observation makes the concept of weak solutions applicable
 to the problem~(\ref{eqn:BPVwithBeta}),
 see Definition~\ref{def:WeakSolution} and Remark~\ref{rem:WeakSolutions}.
\end{rem}


The proof of Theorem~\ref{thm:BVPwithBetSolution}
 is given at the end of this section.

\smallbreak

%
%
%
%
%
%

In order to discuss the notion of M-operators associated with the
boundary value problem~(\ref{eqn:BPVwithBeta})
define the corresponding M-operators as follows.
The solution~$w_z^\varphi := w_z^{0,\varphi}$ is obtained in the closed
form by putting~$f =0 $ in~(\ref{eqn:BPVwithBeta}) and
(\ref{eqn:w_z}):
 \[
   w_z^\varphi =
    (I - z\Ad^{-1})^{-1}\Pi  (\overline{\beta_0 + \beta_1
    M(z)})^{-1}\varphi
    = S_z (\overline{\beta_0 + \beta_1 M(z)})^{-1}
 \varphi
 \]
Vector~$w_z^\varphi$ belongs to the domain of $\Gd$ for any $\varphi
\in E$ and
\[
 \Gd w_z^\varphi   = (\overline{\beta_0 + \beta_1 M(z)})^{-1}\varphi
\]
Hence the operator~$(\overline{\beta_0 + \beta_1 M(z)})^{-1}$
 could be termed ``$(\beta_0\beta_1) $-to-$(I,0)$
 map.''
The notation~``$(I,0)$'' reflects equalities $\beta_0 = I$,
$\beta_1 =0$ that  correspond to the condition~$\Gd w = 0$
 in~(\ref{eqn:BPVwithBeta}).
At the same time the inclusion~$w_z^\varphi \in \Dom(\Gn)$ needs not
be valid for an arbitrary $\varphi\in E$.
However, if there exists a set of~$\varphi \in E$ such that
 vectors $(\overline{\beta_0 + \beta_1 M(z)})^{-1}\varphi$
 lie in~$\Dom(\Lambda)$, similar arguments lead to the
 definition of $(\beta_0\beta_1) $-to-$(0,I)$ map
$
 M(z)(\overline{\beta_0 + \beta_1
 M(z)})^{-1}
 $
 that may be unbounded and even non-densely defined as an operator in~$E$.
%
%
%


%
%
This argumentation is easily extendable to the definition of
M-operators as $(\beta_0\beta_1)$-to-$(\alpha_0\alpha_1)$-maps,
 where $\alpha_0$, $\alpha_1$ is another pair of ``boundary operators''
 from  the boundary
 condition~$(\alpha_0\Gd + \alpha_1 \Gn) u = \psi$.
Such a map is formally given by the ``linear-fractional
transformation with operator coefficients'' $(\alpha_0 + \alpha_1
M(z) )(\overline{\beta_0 + \beta_1 M(z)})^{-1}$.
The precise meaning of this formula needs to be clarified in each
particular case at hand.
 Operator transformations of this kind (with $z$\nobreakdash-dependent
 coefficients) are typical in the systems
 theory where  M\nobreakdash-functions are realized as transfer functions
 of linear systems, see~\cite{Ryz2}, \cite{Staff}.
For the cases when  $\Ker(\overline{\beta_0 + \beta_1 M(z)}) \neq \{0\}$
 relevant results are given by
  the boundary triplets approach in~\cite{DHMS1}, \cite{DHMS3}, \cite{DM2}
   in terms of linear boundary relations in
    Hilbert and Krein spaces.

\bigskip


The section concludes with  the proof of
Theorem~\ref{thm:BVPwithBetSolution}.


\begin{proof}
As clarified in Remark~\ref{rem:B0B01MClosed} and
 Remark~\ref{rem:WeakSolForBetaBVP},
 operators
 $\beta_0 + \beta_1 M(z)$
 are closed  on~$\Dom(\mathscr B )$
 simultaneously  for all $z\in\rho(\Ad)$ and
 in accordance with
 Theorem~\ref{thm:BVPSolution} the vector~$w_z^{f,\varphi}$
  from~(\ref{eqn:w_z}),
 (\ref{eqn:Psi_z})
 is a solution to the system~(\ref{def:BVP}) with $\varphi$
 replaced by~$\Psi_z^{f,\varphi}$.
 In particular, Theorem~\ref{thm:BVPSolution}
 implies that $\Gd w_z^{f,\varphi} = \Psi_z^{f,\varphi}$
 and the
 solution~$w_z^{f,\varphi} \in \Ker(A -zI)$
  is unique.
%


Assume the vector~$\Psi_z^{f,\varphi}$ defined by~(\ref{eqn:Psi_z})
belongs to $\Dom(\Lambda)$ so that $w_z^{f,\varphi} \in
 \Dom(\Gn)$.
 Then
 \[
 \Gn w_z^{f,\varphi} = \Gn (\Ad -zI)^{-1} f + \Gn (I -
z\Ad^{-1})^{-1}\Pi\Psi_z^{f,\varphi}
 = \Pi^* (I - z\Ad^{-1})^{-1} f + M(z)\Psi_z^{f,\varphi}
 \]
Therefore
 \[
  \begin{aligned}
 (\beta_0\Gd + \beta_1\Gn)w_z^{f,\varphi}  & = (\beta_0
 + \beta_1 M(z)) \Psi_z^{f,\varphi} + \beta_1 \Pi^* (I - z\Ad^{-1})^{-1} f
 \\
 &  = \varphi - \beta_1\Pi^*(I - z\Ad^{-1})^{-1}f
  + \beta_1 \Pi^* (I - z\Ad^{-1})^{-1} f = \varphi
  \end{aligned}
 \]
Hence both equations~(\ref{eqn:BPVwithBeta}) are satisfied if
$\Psi_z^{f,\varphi} \in \Dom(\Lambda)$.
%
%

In the general case when $\Psi_z^{f,\varphi} \in \Dom(\mathscr B)$
the vector~$w_z^{f,\varphi}$ from (\ref{eqn:w_z}) belongs to
 $\mathscr H_\mathscr B$ and therefore the expression
 $(\beta_0\Gd + \beta_1\Gn)w_z^{f,\varphi}$ is well defined
  in accordance with Lemma~\ref{lem:B0G0B1G1Boundedness}.
 We need only show that it is equal to~$\varphi$, as required by
 the second equation in
 (\ref{eqn:BPVwithBeta}).
 Consider the sequence $\Psi_n \in \Dom(\Lambda) $, $n = 0 ,1,\dots $
 such that $\Psi_n \to \Psi_z^{f,\varphi}$ in the graph norm of operator~$\mathscr
 B$.
 Then vectors $w_n \in \mathscr D$ defined by~(\ref{eqn:w_z})
 with $\Psi_z^{f,\varphi}$ replaced by $\Psi_n$
 converge to $w_z^{f,\varphi}$ in the metric of~$\mathscr H_\mathscr B$
 as $n \to \infty$.
Due to the boundedness of expression~$(\beta_0\Gd + \beta_1\Gn)$ as
an operator from $ \mathscr H_\mathscr B$ to $ E$,
\begin{equation}\label{eqn:aux1}
 \lim_{n\to\infty} (\beta_0\Gd + \beta_1\Gn)w_n = (\beta_0\Gd +
 \beta_1\Gn)w_z^{f,\varphi}
\end{equation}
From the other side,
\[
 \begin{aligned}
 (\beta_0\Gd + \beta_1\Gn)w_n & =
(\beta_0\Gd + \beta_1\Gn) [(\Ad -zI)^{-1}f + (I
-z\Ad^{-1})^{-1}\Pi\Psi_n]
\\
 &= \beta_0\Psi_n + \beta_1 \Gn (\Ad -zI)^{-1} f + \beta_1
 M(z)\Psi_n
 \end{aligned}
\]
Since $\beta_0\Psi_n + \beta_1
 M(z)\Psi_n \to (\overline{\beta_0 + \beta_1 M(z)})\Psi_z^{f,\varphi}$
 as $n\to\infty$,
 we see that
\[
 \lim_{n\to\infty}(\beta_0\Gd + \beta_1\Gn)w_n =
  (\overline{\beta_0 + \beta_1 M(z)})\Psi_z^{f,\varphi} +
 \beta_1 \Gn (\Ad -zI)^{-1} f
\]
Direct substitution of~$\Psi_z^{f,\varphi}$ from~(\ref{eqn:Psi_z})
yields $ (\beta_0\Gd + \beta_1\Gn)w_n \to \varphi $ as $n\to\infty$.
In accordance with~(\ref{eqn:aux1}),  the equality $(\beta_0\Gd +
\beta_1\Gn)w_z^{f,\varphi} = \varphi$ follows.
\qed
\end{proof}


\section{Linear Operators of Boundary Value Problems}\label{EXTENT}

Let $A_{00}$ be the minimal operator
 defined as a restriction of $A$
 to the set of elements $u\in \mathscr D$
 satisfying conditions~$\Gd u = \Gn u =0 $.
This section is concerned with extensions of~$A_{00}$ to operators
 corresponding to ``boundary conditions'' of the form
$(\beta_0\Gd +\beta_1\Gn) u = 0$.
These operators are first defined via their resolvents
 given by a version of Krein's resolvent
 formula~\cite{Krein2}, \cite{LangText}.
More conventional definitions via boundary conditions
 are provided in terms of extensions of~$A_{00}$.
The groundwork for the study is laid down in
 Theorem~\ref{thm:BVPwithBetSolution}.


\begin{defn}
 Let $A_{00}$ be the restriction of~$\Ad$ to the linear set
\[
\Dom( {A_{00}}) = \Ker({\Gd})\cap \Ker ({\Gn}) = \Dom(\Ad)\cap
\Ker(\Gn),
\]
that is, $A_{00} = \left. A\right|_{\Dom({A_{00}})}$.
We call $A_{00}$ the\textit{ minimal operator}.
\end{defn}
%
%

%
%
%

The next characterization of~$\Dom(A_{00})$ is more universal
since it does not involve the map~$\Gn$.
Recall that $\Ker(A) = \Ran (\Pi)$ by definition of~$A$.

\begin{rem}\label{rem:DomA00RanA00}
The domain~$\Dom(A_{00})$ is described as follows
\[
\Dom(A_{00}) = \left\{ u \in \Dom(\Ad) \mid  \Ad u \perp
\Ran(\Pi)\right \} =  \Ad^{-1}  \left(\Ran(\Pi)^\perp\right)
\]
where $\Ran(\Pi)^\perp$ is the orthogonal complement to the range of
$\Pi$.
The range of $A_{00}$ is closed in $H$ and coincides with the
subspace~$\Ran(\Pi)^\perp = H \ominus \overline{\Ker (A) }$.
\end{rem}
%

\begin{proof}
Indeed, if $u\in \Dom(\Ad)$ then $u = \Ad^{-1} f$ with some $f\in
H$.
The condition $\Gn u =0 $ means that $ \Gn \Ad^{-1} f =0 $, or $f\in
\Ker (\Pi^*)$ (since $\Gn \Ad^{-1} = \Pi^* $), which is equivalent
to $f \perp \Ran(\Pi)$.
 The second statement holds because
$A_{00}\Ad^{-1} \Ran(\Pi)^\perp = \Ran (\Pi)^\perp = H \ominus
\overline{\Ker{A}}$.
\qed
\end{proof}
%
%


\begin{rem}
The equality $\Dom(A_{00}) = \Ad^{-1} \Ran(\Pi)^\perp$
 shows in particular that the operator~$A_{00}$ does
 not depend on any given
 choice of $\Lambda$.
Moreover, $A_{00}$ is symmetric but need not be densely defined.
The operator~$\Ad$ is a selfadjoint extension of $A_{00}$
contained in~$A$.
\end{rem}

%
%
%
%
Relations (\ref{eqn:Psi_z}) and
(\ref{eqn:w_z}) offer a rather natural way to define the resolvent
of an operator associated with the ``boundary
con\-di\-ti\-on''~$(\beta_0 \Gd + \beta_1 \Gn) u =0$.
By putting $\varphi =0$ and inserting (\ref{eqn:Psi_z}) into
(\ref{eqn:w_z}) a suitable candidate for the role of
resolvent is obtained:
 \begin{multline}\label{eqn:ResolventAbb}
   \mathscr  R_{\beta_0,\beta_1}(z)   =
\\
  (\Ad  - z I)^{-1}
 -  (I - z\Ad^{-1})^{-1} \Pi
 (\overline{\beta_0 + \beta_1 M(z)})^{-1}
 \beta_1\Pi^*(I - z\Ad^{-1})^{-1}
 \end{multline}
As will be shown, the operator function~(\ref{eqn:ResolventAbb}) is
 indeed the resolvent of some closed linear operator~$\Abb$ in~$H$ whose
 domain~$\Dom(\Abb)$ coincides with the set $\Ker (\beta_0 \Gd + \beta_1 \Gn)$.
Assuming the conditions of Theorem~\ref{thm:BVPwithBetSolution} are
satisfied, denote
\[
   Q_{\beta_0,\beta_1}(z) = - (\overline{\beta_0 + \beta_1 M(z)})^{-1}
 \beta_1
\]
The operator-function~$ Q_{\beta_0,\beta_1}(z)$ is analytic and
bounded as long as $z\in\rho(\Ad)$ satisfies conditions of
Theorem~\ref{thm:BVPwithBetSolution}.
The expression~(\ref{eqn:ResolventAbb}) for~$\mathscr R_{\beta_0,
\beta_1}(z)$ takes the form
\begin{equation}\label{eqn:ResolbventAbbShort}
 \mathscr R_{\beta_0, \beta_1}(z) = R_z + S_z Q_{\beta_0,\beta_1}(z)S_{\bar
 z}^{\,*}
\end{equation}
where $R_z = (\Ad -zI)^{-1}$ is the resolvent of~$\Ad$ and  $S_z =
 (I - z\Ad^{-1})^{-1}\Pi$ is the solution operator.
For simplicity, the
 indices in $Q_{\beta_0,\beta_1}$ will be omitted and the  notation $Q(z)$ will be used for $Q_{\beta_0,\beta_1}(z)$
 when it does not lead to confusion.
 An important analytical property of~$Q(z)$ is formulated in the next Lemma.
%
%

%

%
%
\begin{lem}\label{lem:Qz-Qzeta}
For $z,\zeta \in \rho(\Ad)$ satisfying
 assumptions of Theorem~\ref{thm:BVPwithBetSolution}
  the equality holds
\[
 Q(z) - Q(\zeta) = (z - \zeta)Q(z)S_{\bar z}^{\,*}S_\zeta Q(\zeta)
\]
\end{lem}
%
%

%

%
%
\begin{proof}
 By virtue of formula~(3) from
 Theorem~\ref{thm:PropertiesOfMz}
 we have for $\varphi \in \Dom(\Lambda)$
\[
  \begin{aligned}
 (z & - \zeta)  \beta_1
   S_{\bar z}^{\,*}S_\zeta
\varphi
  =
 \beta_1 \left[M(z) - M(\zeta)\right]\varphi
 \\
 &  = (\overline{\beta_0 + \beta_1 M(z)})\varphi
    - (\overline{\beta_0 + \beta_1 M(\zeta)})\varphi
\\
 & = (\overline{\beta_0 + \beta_1 M(z)}) \left[ (\overline{\beta_0 + \beta_1 M(\zeta)})^{-1} -
 (\overline{\beta_0 + \beta_1 M(z)})^{-1} \right ]
  (\overline{\beta_0 + \beta_1 M(\zeta)})\varphi
  \end{aligned}
\]
Therefore
\[
  \begin{aligned}
(z   - \zeta)Q(z)S_{\bar z}^{\,*}S_\zeta Q(\zeta)  
& = \left[ (\overline{\beta_0 + \beta_1 M(\zeta)})^{-1} -
 (\overline{\beta_0 + \beta_1 M(z)})^{-1} \right ]\beta_1 
\\
&= Q(z) - Q(\zeta)
  \end{aligned}
\]
as stated.  \qed
\end{proof}


The main Theorem of this section reads as follows.


\begin{thm}\label{thm:Extensions}
Assume $z\in\rho(\Ad)$ is such that the closed
 operator $\overline{\beta_0 + \beta_1 M(z)}$ defined on~$\Dom(\mathscr B)$ is
 boundedly invertible in the space~$E$.
  Then the
  operator function~$ \mathscr  R_{\beta_0,\beta_1}(z)$
  defined by~(\ref{eqn:ResolventAbb})
  is the resolvent
  of a closed densely defined operator~$\Abb$ in~$H$.
 For $\Abb$ the inclusions are valid
\begin{equation}\label{eqn:Inclusions}
 A_{00}\subset\Abb \subset A, \quad
\end{equation}
The domain of~$\Abb$ satisfies
\begin{equation}\label{eqn:DomainAbb}
 \Dom(\Abb) = \{ u \in \mathscr H_\mathscr
 B \mid (\beta_0\Gd + \beta_1\Gn) u = 0 \}
 = \Ker(\beta_0\Gd + \beta_1\Gn)
\end{equation}
In addition,
\begin{equation}\label{eqn:GdResolventAbb}
 \Gd ( \Abb - zI)^{-1} =  Q(z) \Gn (\Ad - zI)^{-1}
\end{equation}
and the resolvent identity holds:
\begin{equation}\label{eqn:ResolventIdentityAdAbb}
 ( \Abb - zI)^{-1}  - (\Ad  - zI)^{-1} =
 \left[\Gn (\Ad  - \bar z I)^{-1} \right]^* \Gd ( \Abb - zI)^{-1}
\end{equation}
\end{thm}


\begin{proof}
Operator function $\mathscr R(z) = \mathscr R_{\beta_0,\beta_1}(z)$
is bounded and analytic for suitable $z\in \Complex$.
 To show that $\mathscr R(z)$
 is a resolvent, we need to check three
 conditions~\cite{Kato2}.
They are:
 1) $\Ker(\mathscr R (z)) = \{0\}$,
 2) $\Ran(\mathscr R (z) )$ is dense in $H$,
 and 3) the function $\mathscr R(z)$
 satisfies the first resolvent equation
\begin{equation}\label{eqn:ResolvenEquation}
\mathscr R(z) -\mathscr R(\zeta)
 = (z - \zeta)\mathscr R(z)
 \mathscr R(\zeta)
\end{equation}
 The equality~$\Ker(\mathscr R(z)) = \{0\}$
 follows directly from the last statement of Theorem~\ref{thm:BVPSolution}.
The same argument applied to $[\mathscr R(z)]^*$ in conjunction with
boundedness of $Q(z)$ and equality~$\Ker ([\mathscr R(z)]^*) = H
\ominus {\Ran(\mathscr R(z))}$ shows that the range of $\mathscr
R(z)$ is dense in $H$.
%


We shall verify the resolvent identity for $\mathscr R(\cdot)$
 written in  simplified notation~(\ref{eqn:ResolbventAbbShort}).
\[
 \begin{aligned}
\mathscr R(z) \mathscr R(\zeta)    &
 = \big(R_z + S_z Q(z) S_{\bar z}^{\,*}\big)\times
 \big(R_\zeta + S_\zeta Q(\zeta) S_{\bar \zeta}^{\,*}\big)
 \\
 &
  =   R_z R_\zeta +  R_z S_\zeta Q(\zeta)
 S_{\bar \zeta}^{\,*}
  + S_z Q(z)  S_{\bar z}^{\,*} R_\zeta
   + S_z Q(z) S_{\bar z}^{\,*}  S_\zeta Q(\zeta) S_{\bar \zeta}^{\,*}
 \end{aligned}
 \]
Multiplying by~$(z-\zeta)$ and noticing that~$R_z - R_\zeta =
(z-\zeta)R_z R_\zeta$ due to the resolvent identity for $\Ad$,  the identity~(\ref{eqn:ResolvenEquation})
is rewritten as
 \begin{multline*}
   S_z Q(z) S_{\bar z}^{\,*}   - S_\zeta Q(\zeta) S_{\bar \zeta}^{\,*}
   =
(z  -\zeta) \left[ R_z S_\zeta Q(\zeta) S_{\bar \zeta}^{\,*}
  + S_z Q(z) S_{\bar z}^{\,*} R_\zeta\right]
\\
     + (z-\zeta)
   S_z Q(z) S_{\bar z}^{\,*}   S_\zeta Q(\zeta) S_{\bar \zeta}^{\,*}
 \end{multline*}
By virtue of (\ref{eqn:RzMinusRzeta}), its adjoint, and
 Lemma~\ref{lem:Qz-Qzeta}
 the right hand side of this equality is
 \[
 (S_z - S_\zeta) Q(\zeta)S_{\bar \zeta}^{\,*}  +
S_z Q(z)  ( S_{\bar z}^{\,*} - S_{\bar \zeta}^{\,*})
 + S_z  (  Q(z) - Q(\zeta) ) S_{\bar \zeta}^{\,*}
 \]
which coincides with the left hand side.
The existence of a closed densely defined
operator~$A_{\beta_0,\beta_1}$ with the resolvent~$(\Abb -zI)^{-1}$
defined by (\ref{eqn:ResolventAbb}) thereby is proven.

%
%
%
%
%
%

Turning to the proof of (\ref{eqn:Inclusions}),
notice that in accordance with~(\ref{eqn:ResolbventAbbShort}) the range of
$(\Abb -zI)^{-1}$ is
 contained in $\Dom(A)$
and  since $ S_z f  \in \Ker(A -zI)$ for  $f\in H$,
 \[
  (A -zI) (\Abb -zI)^{-1}f =
  (A -zI)\left( R_z + S_z Q(z) S_{\bar z}^{\,*}\right) f =
(A -zI) R_z f = f
 \]
Hence $A g = \Abb g$ for $g \in \Dom(\Abb)$, which means $\Abb
\subset A$.
%

%
%
%
%

 To prove the
 inclusion~$\Dom(\Abb) \subset \Ker(\beta_0 \Gd + \beta_1 \Gn)$
 in (\ref{eqn:DomainAbb})
 note that, as follows from (\ref{eqn:Psi_z}) with $\varphi =0$,
  the vector~$w_z^f = (\Abb -zI)^{-1} f$ is represented as
$
(\Abb -zI)^{-1} f =
  R_z f + S_z \Psi_z^{f}
$
 for each element~$f\in H$,
 where $\Psi_z^{f} =   Q(z)S_{\bar z}^{\,*} f \in
\Dom(\mathscr B)$.
Therefore  $w_z^f \in \mathscr H_\mathscr B$  and
 $(\beta_0 \Gd + \beta_1 \Gn)w_z^f = 0$
 by
Theorem~\ref{thm:BVPwithBetSolution} with $\varphi =0$.
Hence $\Dom(\Abb)$ in included into $\Ker(\beta_0 \Gd + \beta_1
\Gn)$.
%


In order to prove the inverse inclusion first
 consider $u \in \mathscr D$ in the form
 $ u = R_z f + S_z \varphi \in \Ker(\beta_0 \Gd + \beta_1\Gn)$
 with $f \in H$ and $\varphi \in \Dom(\Lambda)$.
Then $u \in\Dom(\Gd)\cap \Dom(\Gn)$ and the operator sum~$\beta_0 \Gd + \beta_1\Gn$
 can be calculated for the element~$u$ termwise, i.~e.
$(\beta_0 \Gd + \beta_1\Gn) u   =  \beta_0 \Gd u  + \beta_1\Gn u  $,
\begin{equation}\label{eqn:Aux}
\begin{aligned}
(\beta_0 \Gd + \beta_1\Gn) u  & =
(\beta_0 \Gd + \beta_1\Gn) (R_z f + S_z \varphi)
\\
&= \beta_0\Gd S_z\varphi + \beta_1 \Gn (\Ad -zI)^{-1} f  + \beta_1\Gn S_z\varphi
\\
 & = \beta_1\Pi^* (I -z\Ad^{-1})^{-1}f + (\beta_0 + \beta_1 M(z)) \varphi
\end{aligned}
\end{equation}
If $u \in \Ker(\beta_0 \Gd + \beta_1\Gn)$, the left hand side of~(\ref{eqn:Aux}) equals zero, so that
\[
\varphi = - (\overline{\beta_0 + \beta_1 M(z)})^{-1}\beta_1\Pi^* (I -z\Ad^{-1})^{-1}f =
 Q_{\beta_0,\beta_1}(z) S_{\bar z}^* f
\]
by virtue of invertibility of $\beta_0 + \beta_1 M(z)$.
Thus, vector~$u = R_z f + S_z \varphi$ due to (\ref{eqn:ResolbventAbbShort}) is
\[
 u = (R_z + S_z  Q_{\beta_0,\beta_1}(z) S_{\bar z}^* )f = \mathscr R_{\beta_0,\beta_1}(z) f
\]
Since $\mathscr R_{\beta_0,\beta_1}(z)$ is the resolvent of$\Abb$,
we have $u \in \Dom(\Abb)$.
%

Consider now
 the general case of element~$v = R_z f + S_z\varphi \in \mathscr H_\mathscr B$, $f \in H$, $\varphi \in \Dom(\mathscr B)$
so that
 $v \notin \mathscr D $  and the operator sum~$\beta_0\Gd + \beta_1 \Gn$ calculated on~$v$
 cannot be computed termwise.
Since the set~$\mathscr D $ is dense in the Hilbert space~$\mathscr H_\mathscr B$, see Definition~\ref{defn:HB}
  and Lemma~\ref{lem:B0G0B1G1Boundedness},
  there exists a sequence~$v_n = R_z f_n + S_z \varphi_n$ with $f_n \in H$ and $\varphi_n \in \Dom(\Lambda)$
 converging to~$v$ in $\mathscr H_\mathscr B$ as $n\to \infty$.
 It  means  in particular that $\varphi_n \to \varphi $ and
 $(\overline{\beta_0 + \beta_1 M(z)}) \varphi_n \to (\overline{\beta_0 + \beta_1 M(z)}) \varphi$
 for $n \to \infty$.
Because~$\beta_0\Gd + \beta_1 \Gn$ is bounded as an operator from ~$\mathscr H_\mathscr B$ to $E$ by virtue of Lemma~\ref{lem:B0G0B1G1Boundedness},
 we have
\begin{equation}\label{eqn:Aux2}
(\beta_0\Gd + \beta_1 \Gn) v_n \to (\beta_0\Gd + \beta_1 \Gn)v, \qquad n \to \infty
\end{equation}
Expression for $(\beta_0\Gd + \beta_1 \Gn)v_n$ follows from~(\ref{eqn:Aux})
\[
 (\beta_0\Gd + \beta_1 \Gn)v_n =
\beta_1\Pi^* (I -z\Ad^{-1})^{-1}f_n + (\beta_0 + \beta_1 M(z)) \varphi_n
\]
Because of the boundedness of $\beta_1\Pi^* (I -z\Ad^{-1})^{-1}$ and closability of $\beta_0 + \beta_1 M(z)$
 this  leads to
\[
 (\beta_0\Gd + \beta_1 \Gn)v_n \to
\beta_1\Pi^* (I -z\Ad^{-1})^{-1}f + (\overline{\beta_0 + \beta_1 M(z)}) \varphi, \qquad n \to \infty
\]
Comparison with~(\ref{eqn:Aux2}) gives
\[
(\beta_0\Gd + \beta_1 \Gn)v = \beta_1\Pi^* (I -z\Ad^{-1})^{-1}f + (\overline{\beta_0 + \beta_1 M(z)}) \varphi
\]
Therefore if $v \in \Ker(\beta_0 \Gd +\beta_1 \Gn) $, then under conditions of Theorem
\[
 \varphi = - (\overline{\beta_0 + \beta_1 M(z)})^{-1}\beta_1\Pi^* (I -z\Ad^{-1})^{-1}f = Q_{\beta_0,\beta_1}(z)S_{\bar z}^* f
\]
Hence, $v = (R_z f + S_zQ_{\beta_0,\beta_1}(z)S_{\bar z}^*) f= \mathscr R_{\beta_0,\beta_1} (z) f$
and $v \in \Dom(\Abb)$.
%

%
%
%


%
%
To prove that~$A_{00}\subset \Abb$ in (\ref{eqn:Inclusions})
  we need to show that any vector~$u$
  from $\Dom(A_{00})$ belongs
  to $\Dom(\Abb) $, in other words,
  can be
  represented in the form~$ u = (\Abb - zI)^{-1} f $
  with some $f\in H$.
 Suppose $u\in \Dom(A_{00})$
 and
 let us choose $f = (A_{00} - zI)u$.
 Then $ f = (\Ad - zI)u$  because $A_{00}\subset \Ad$ 
and
  \begin{multline*}
 (\Abb - zI)^{-1} f  = (\Abb - zI)^{-1} (\Ad - zI)u
 \\
   = \left( R_z + S_z Q(z) S_{\bar z}^* \right) (\Ad - zI)u
  = u + S_z Q(z) S_{\bar z}^* (\Ad - zI)u = u
  \end{multline*}
 The last equality holds due to identities
 \[
 S_{\bar z}^* (\Ad - zI)u = \Pi^* (I - z\Ad^{-1})^{-1} (\Ad - zI)u = \Gn
 u,
 \quad u \in \mathscr D
 \]
 and $\Gn u = 0 $ for $u\in \Dom(A_{00})$.
All claims~(\ref{eqn:Inclusions}) and (\ref{eqn:DomainAbb}) are proven.


Finally, in the notation above the
formula~(\ref{eqn:GdResolventAbb})
 is equivalent to the already established relation $\Gd w_z^f = \Psi_z^f$.
 The resolvent identity~(\ref{eqn:ResolventIdentityAdAbb})
 is obtained
 from (\ref{eqn:ResolventAbb})
 by (\ref{eqn:GdResolventAbb})
 and  equality~$\Gn (\Ad -zI)^{-1} = \Pi^* (I
-z\Ad^{-1})^{-1}$.
\qed
\end{proof}


\begin{rem}
  Equalities~(\ref{eqn:ResolventAbb}) and (\ref{eqn:ResolventIdentityAdAbb})
  are correspondingly
  Krein's formula and Hilbert resolvent identity
  for $\Ad$ and $\Abb$.
\end{rem}

\begin{rem}
  Let $\widetilde\beta_0$ and $\widetilde\beta_1$
  be two linear operators with the same properties as $\beta_0$ and $\beta_1$
   in Theorem~\ref{thm:Extensions}.
 A natural question arises as to whether
 the boundary conditions~$(\widetilde\beta_0\Gd + \widetilde\beta_1\Gn)u = 0 $
  define the same operator as the con\-di\-ti\-ons~$(\beta_0\Gd + \beta_1\Gn)u = 0 $
   discussed in Theorem.
 One obvious answer is that when $\beta_0 = C \widetilde \beta_0$  and
 $\beta_1 = C \widetilde \beta_1$ with some operator $C$ such that~$\Ker(C) = \{0\}$
 then the equality~$A_{\beta_0,\beta_1} = A_{\widetilde\beta_0,\widetilde\beta_1}$ holds
 because the null sets $\Ker(\beta_0\Gd + \beta_1\Gn)$ and  $\Ker(\widetilde\beta_0\Gd + \widetilde\beta_1\Gn)$
 are equal.
 Necessary and sufficient condition follows from the formula~(\ref{eqn:ResolventAbb}).
 Namely, the identity
 $
 \Pi Q_{\beta_0,\beta_1}(z)\Pi^* =  \Pi Q_{\widetilde\beta_0,\widetilde\beta_1}(z)\Pi^*
$
for $z$ in a (non-empty) domain of the complex plane is equivalent to the identity of resolvents of
$A_{\beta_0,\beta_1}$ and $A_{\widetilde\beta_0,\widetilde\beta_1}$, 
 thus to the equality~$A_{\beta_0,\beta_1} = A_{\widetilde\beta_0,\widetilde\beta_1}$.
\end{rem}


\begin{cor}\label{cor:Abb-1}
 Assume the operator~$\mathscr B = \overline{\beta_0 + \beta_1 \Lambda}$
 is boundedly invertible in $E$.
 Then $\Abb$ is boundedly invertible in $H$,
\[
 A_{\beta_0, \beta_1}^{-1} = \Ad^{-1} - \Pi
 (\overline{\beta_0 + \beta_1 \Lambda})^{-1}\beta_1\Pi^*
 = \Ad^{-1} + \Pi Q(0)\Pi^*,
\]
and
 $Q(z)$ has the representation
\[
 Q(z) = Q(0) + z Q(0) \Pi^* (I - z A_{\beta_0, \beta_1}^{-1})^{-1}\Pi Q(0)
\]
at least in a small neighborhood of $z =0$.
\end{cor}


\begin{proof}
Noting that
 $Q(0) = - (\overline{\beta_0 + \beta_1 \Lambda})^{-1}\beta_1$ is
 bounded,
 invertibility of~$\Abb$ and the formula for~$A_{\beta_0, \beta_1}^{-1}$
 follow directly from~(\ref{eqn:ResolventAbb}) or
(\ref{eqn:ResolbventAbbShort}).
  Existence of $Q(z) = - (\overline{\beta_0 + \beta_1 M(z)})^{-1}\beta_1$
  for small $|z|$ results from
  analyticity and invertibility of $\overline{\beta_0 + \beta_1 M(z)}$ at $z =0$.
Lemma~\ref{lem:Qz-Qzeta} with $\zeta =0 $ yields
\begin{equation}\label{eqn:Qz=Q0+}
 Q(z) = Q(0) + z Q(z) S_{\bar z}^{\,*} S_0 Q(0)
\end{equation}
Observe now that $Q(z) S_{\bar z}^{\,*} = Q(z) \Gn (\Ad - zI)^{-1}$,
thus according to~(\ref{eqn:GdResolventAbb}),
\[
Q(z) S_{\bar z}^{\,*}  = \Gd (\Abb - zI)^{-1} = \Gd
 A_{\beta_0, \beta_1}^{-1} ( I - zA_{\beta_0, \beta_1}^{-1})^{-1}
\]
 Formula~(\ref{eqn:GdResolventAbb}) for $z =0$  gives $ \Gd
 A_{\beta_0, \beta_1}^{-1} = Q(0) \Gn \Ad^{-1} = Q(0)\Pi^*$
 so that
\[
 Q(z)S_{\bar z}^{\,*} = Q(0) \Pi^* ( I - zA_{\beta_0,
 \beta_1}^{-1})^{-1}
 \]
In combination with $S_0 Q(0) = \Pi Q(0)$
 the expression~(\ref{eqn:Qz=Q0+}) yields
 the required representation for~$Q(z)$.
\qed
\end{proof}

\begin{cor}\label{cor:AbbSelfadjoint}
  Assume conditions of
   Corollary~\ref{cor:Abb-1} are satisfied, operators~$\beta_0$, $
   \beta_1$ and $\Lambda$ are bounded, and $\beta_0 \beta_1^*$
   is selfadjoint.
 Then $\Abb$ is selfadjoint.
\end{cor}


\begin{proof}

 Since $A_{\beta_0, \beta_1}^{-1} -  (A_{\beta_0, \beta_1}^{-1})^* =
  \Pi \left[(\beta_0 + \beta_1\Lambda )^{-1} \beta_1 - \beta_1^*
  (\beta_0^* + \Lambda \beta_1^* )^{-1}\right]\Pi^*
 $
\[
   = \Pi (\beta_0 + \beta_1\Lambda )^{-1}
  \left[
    \beta_1 (\beta_0^* + \Lambda \beta_1^* ) -
    (\beta_0 + \beta_1\Lambda )\beta_1^*
  \right]
  (\beta_0^* + \Lambda \beta_1^* )^{-1} \Pi^* = 0
\]
 under assumption~$\beta_1 \beta_0^* = \beta_0 \beta_1^*$,
  the operator~$\Abb$ is an (unbounded) inverse of the bounded selfadjoint
  operator.
\qed
\end{proof}


A special case of operator~$\Abb$ in Theorem~\ref{thm:Extensions}
 with $\beta_0 = 0 $,
 $\beta_1 = I$
 is of particular interest.
It can be seen as an abstract analogue  of the Laplacian
 with Neumann boundary condition from
 Section~\ref{DNMAP}.
 Note that in this case
   $Q(z) = - (M(z))^{-1}$ and $Q(0) = -\Lambda^{-1}$.


\begin{cor}\label{cor:aux1}
 Suppose $\Lambda$ is boundedly invertible.
  Then operator~$\An$ defined as a restriction of~$A$
  to the set $\Dom(A_1) = \{ u \in \mathscr D\mid \Gn u =0\}$
  is selfadjoint and boundedly invertible.
   For $ z \in\rho(\Ad)\cap \rho(\An)$
\begin{equation}\label{eqn:cor55}
 \begin{aligned}
  (\An -zI)^{-1} &= (\Ad -zI)^{-1} - (I - z\Ad^{-1})^{-1}\Pi
   (M(z))^{-1} \Pi^*(I - z\Ad^{-1})^{-1}
  \\
  \text{ where }\,\, (M(z))^{-1} & = \Lambda^{-1} - z \Lambda^{-1} \Pi^* (I - z
  \An^{-1})^{-1} \Pi \Lambda^{-1}, \quad
    z \in \rho(\An).
\end{aligned}
\end{equation}
Moreover, for $ z \in\rho(\Ad)\cap \rho(\An)$
\begin{equation}\label{eqn:cor55bis}
  (\An -zI)^{-1} = (\Ad -zI)^{-1} - (I - z\An^{-1})^{-1} \pi M(z)\pi^* (I - z\An^{-1})^{-1},
\end{equation}
where $\pi = (\Gd \An^{-1})^*$ is bounded with $\Ran(\pi^*) \subset \Dom(\Lambda)$.

 In particular,
  $\An^{-1} =  \Ad^{-1} - \Pi \Lambda^{-1} \Pi^* = \Ad^{-1} - \pi \Lambda \pi^*$ where 
both  $ \Pi \Lambda^{-1} \Pi^* $ and $ \pi \Lambda \pi^*$ are bounded operators.
\end{cor}


\begin{proof}
The first equality~(\ref{eqn:cor55})
 follows directly from~(\ref{eqn:ResolventAbb})
 and Theorem~\ref{thm:Extensions}.
 Selfadjointness of~$\An$ is a consequence of representation of $\An^{-1}$
 as a sum of two bounded selfadjoint operators.
Because $\An^{-1}$ is bounded,  the analytic operator function
 $(M(z))^{-1}$  from (\ref{eqn:cor55})
  can be analytically continued from a neighborhood
 of the origin~$z =0$ to all $z \in \rho(\An)$.
The alternative representation~(\ref{eqn:cor55bis})
 is obtained from~(\ref{eqn:cor55}) with the help of
 equalities~(\ref{eqn:GdResolventAbb}) and~(\ref{eqn:ResolventIdentityAdAbb}).
Boundedness of the operator function~$\pi M(z)\pi^*$,
$z\in\rho(\Ad)$, or equivalently of the operator~$\pi\Lambda\pi^*$,
  is ensured by the calculations
 \[
   \pi^* = \Gd \An^{-1} = \Gd (\Ad^{-1} - \Pi \Lambda^{-1}\Pi^*)
    = -\Gd \Pi \Lambda^{-1}\Pi^* = -\Lambda^{-1} \Pi^*,
 \]
so that $\Lambda \pi^* = - \Pi^*$.
This equality also follows from~(\ref{eqn:GdResolventAbb})
 with $z =0$.
\qed
\end{proof}


There exists a close relationship between analytical properties of
the
 ope\-ra\-tor-func\-ti\-on~$Q_{\beta_0,\beta_1}(z)$
 and spectral characteristics
 of $\Abb$.
For example, papers~\cite{DM1}, \cite{DM2}  report some
 general results obtained within the boundary triplet based framework
 when ~$\beta_1 =I$, $\beta_0$ is closed, $\Ran(\Gd) = \Ran(\Gn) = E$,
 and therefore~$M(z)$ is bounded.
The next theorem regarding the point spectrum
 of~$\Abb$ renders similar results in the paper's setting.
Much more complicated
 relationships between spectral properties of nonselfadjoint operators and
  their M\nobreakdash-funct\-ions are discussed in~\cite{BHMNW}, \cite{BMNW}.
%


%
%
\begin{thm}\label{thm:discreteSpectrum}
 Assume the operator~$\mathscr B = \overline{\beta_0 + \beta_1 \Lambda}$
 is boundedly invertible.
Then for any $z\in \rho(\Ad)$ the mapping~$\varphi \mapsto S_z
 \varphi$ establishes a one-to-one correspondence between
 $\{\varphi \in \Dom(\mathscr B) \mid (\overline{\beta_0 +
\beta_1 M(z)})\varphi = 0\}$ and $\Ker(\Abb - zI)$.
In particular,
 $\Ker(\overline{\beta_0 + \beta_1 M(z)}) = \{0\}$ is equivalent to
 $\Ker(\Abb -zI) = \{0\}$
 for $z\in \rho(\Ad)$.
\end{thm}


\begin{proof}
 We start with the observation
 that under the Theorem's assumptions
 the operator $Q(0) = - \mathscr B^{-1}\beta_1$
 is bounded.
Hence,
 according to Corollary~\ref{cor:Abb-1},
 $A_{\beta_0, \beta_1}$ is boundedly invertible and
 $ A_{\beta_0, \beta_1}^{-1} = \Ad^{-1} + \Pi Q(0)\Pi^* $.
%


Assume that $(\overline{\beta_0 + \beta_1 M(z)})\varphi  = 0 $ for
some
 $z\in \rho(\Ad)$
 and  $\varphi \in \Dom(\mathscr B)$.
 Let $u  = S_z \varphi $ be the corresponding solution to the
 equation~$(A -zI) u  =0$ satisfying condition $\Gd u = \varphi$.
Then
\[
 0 =
(\overline{\beta_0 + \beta_1 M(z)})\varphi =
  (\overline{\beta_0 + \beta_1 \Lambda})\varphi  + z \beta_1\Pi^*(I -
 z\Ad^{-1})^{-1}\Pi\varphi
\]
and therefore $\varphi$ can be expressed in terms of $u = (I -
 z\Ad^{-1})^{-1}\Pi\varphi$ as follows
\[
 \varphi = - z (\overline{\beta_0 + \beta_1 \Lambda})^{-1}
\beta_1\Pi^*(I -
 z\Ad^{-1})^{-1}\Pi\varphi =
  -z {\mathscr B}^{-1}\beta_1 \Pi^*S_z \varphi
 = z Q(0)\Pi^* u
\]
Owing to identity~$(I - z\Ad^{-1})^{-1} = I + z\Ad^{-1} (I -
z\Ad^{-1})^{-1}$ we  obtain
 \begin{multline*}
 u   = (I - z\Ad^{-1})^{-1}\Pi\varphi
  =  \Pi\varphi + z\Ad^{-1} (I - z\Ad^{-1})^{-1}\Pi  \varphi
  = \Pi \varphi  + z \Ad^{-1}S_z \varphi
\\
 = z \Pi Q(0)\Pi^* u + z \Ad^{-1} u = z (\Ad^{-1} + \Pi Q(0)\Pi^*) u
 = z A_{\beta_0,\beta_1}^{-1} u
 \end{multline*}
 It means inclusion  $u \in \Dom(\Abb)$.
It follows that  $( \Abb  - zI)u  = (A -zI) u=0$
since $\Abb \subset A$ .

%


Suppose now  that $u \in \Ker(\Abb -zI)$ and denote $\varphi =\Gd u
\in E$.
Then~$u$ has the form $u = (I -z\Ad^{-1})^{-1}\Pi \varphi$ because
$\Abb \subset A$
 and therefore $u \in\Ker (A -zI)$.
We need to show that $\varphi$ belongs to
 the domain of $\mathscr B = \overline{\beta_0 + \beta_1 \Lambda}$
 and $\mathscr B \varphi  = - z \beta_1\Pi^* u$.
The  equality $(\Abb -zI) u =0$ implies $(I - z
A_{\beta_0,\beta_1}^{-1})u =0$.
 Hence $ u =  z
A_{\beta_0,\beta_1}^{-1}u =
 z \left(\Ad^{-1} + \Pi Q(0) \Pi^*\right)u$.
 Application of $\Gd$ to both sides  yields
 $\varphi = \Gd u = z Q(0)\Pi^* u $.
 Recall now that $Q(0)   =  - (\overline{\beta_0 + \beta_1 \Lambda})^{-1} \beta_1  =  - \mathscr B^{-1} \beta_1$
 and the required identity $\mathscr B \varphi = - z \beta_1 \Pi^* u$
 follows.
\qed
\end{proof}

%
%
%
%
%
%
%
%
%

The rest of this section is devoted to the special case of
operators~$\Abb$ inspired by the Birman-Krein-Vishik theory of
extensions of positive symmetric operators~\cite{Birman}, \cite{Krein2}, \cite{Vi}.
 Only a simplified version of this theory is considered assuming
 that the extension parameter (operator~$B$ below)
 is densely defined and boundedly invertible
 in the space~$\overline {\Ran(\Pi)}$.
For the general case of the Birman-Krein-Vishik theory  the
  reader is referred,
  apart from the original publications cited above,
  to the work~\cite{Grubb} for the exhaustive treatment
  and to the paper~\cite{AS} for an overview.


Denote $\mathcal H := \overline{\Ran(\Pi)} = \overline{\Ker(A)}$.
Recall that
 according to
Remark~\ref{rem:DomA00RanA00}
 the orthogonal complement of~$\mathcal H$ is the
 subspace
 $\mathcal H^\perp = H \ominus \overline{\Ker(A)} =  \Ran (A_{00})$.
 Let $B$ be a closed
densely defined operator in $\mathcal H$ such that $\Dom(B)\supset
\Pi\Dom(\Lambda)$.
Consider the restriction~$L_B$ of $A$ to the set
\[
\Dom (L_B) = \left\{ \Ad^{-1} ( f_\perp + B h) + h \mid f_\perp \in
\mathcal H^\perp, h \in  \Pi \Dom(\Lambda) \right\}
\]
Since $L_B\subset A$ by definition, we have
\begin{equation}\label{eqn:defLB}
L_B :  \Ad^{-1} ( f_\perp + B h) + h \mapsto f_\perp + B h, \quad
f_\perp \in \mathcal H^\perp,\, h \in \Pi\Dom(\Lambda)
\end{equation}
Clearly  $A_{00}\subset L_B $ because $\Dom(A_{00}) = \Ad^{-1}
 \mathcal H^\perp\subset \Dom(L_B)$.
We would like to show that~$L_B$ is closed and $L_B = \Abb$ for some
$\beta_0$, $\beta_1$.
To simplify the matter, additional conditions of the
  boundedness and invertibility of~$\Pi^* B\Pi$ are imposed in the following Theorem.
%


\begin{thm}\label{thm:LB_and_Abb}
Suppose
 the set
$ B\Pi \Dom(\Lambda)$ is dense in~$\mathcal H $ and the operator
$\Pi^* B\Pi$ is bounded and boundedly invertible in $E$.
Then  the inverse $L_B^{-1}$ exists and
\begin{equation}\label{eqn:aux15}
   L_B^{-1} = \Ad^{-1} + \Pi (\Pi^* B \Pi)^{-1}\Pi^*
\end{equation}
Moreover, $L_B =  A_{\beta_0, \beta_1} $ with
 $\beta_1 = - I_E$ and $\beta_0 = \Lambda +  \Pi^* B \Pi$.
In particular, if  the function
\[
M_B(z) = \Lambda_B + z\Pi^*(I - z\Ad^{-1})^{-1}\Pi,
\qquad \text {with }\quad \Lambda_B = -\Pi^* B \Pi
\]
is boundedly invertible for some~$z\in\rho(\Ad)$, then
$z\in\rho(L_B)$ and
\[
 (L_B - zI)^{-1} = (\Ad - zI)^{-1} -   (I - z\Ad^{-1})^{-1} \Pi
 M_B^{-1}(z) \Pi^* (I - z\Ad^{-1})^{-1}
\]
\end{thm}


\begin{proof}
Formula~(\ref{eqn:aux15})  is verified by direct computations.

\noindent
 Assuming $u = \Ad^{-1}( f_\perp  + B \Pi \varphi) + \Pi \varphi$ with
$f_\perp \in \mathcal H^\perp = \Ker(\Pi^*)$ and $\varphi \in
\Dom(\Lambda)$ we have
\begin{multline*}
    \left(\Ad^{-1} + \Pi (\Pi^* B \Pi)^{-1} \Pi^*
    \right) L_B u =
    \left(\Ad^{-1} + \Pi (\Pi^* B \Pi)^{-1} \Pi^*
    \right)(f_\perp + B \Pi\varphi)
    \\
    = \Ad^{-1} (f_\perp + B \Pi\varphi) +
    \Pi (\Pi^* B \Pi)^{-1} \Pi^* B \Pi\varphi =
    \Ad^{-1} (f_\perp + B \Pi\varphi) + \Pi \varphi = u
\end{multline*}
From the other side, consider $f\in H$ in the form~$f = f_\perp  +
B\Pi \varphi$ with $f_\perp \in \mathscr H^\perp$, $\varphi \in
\Dom(\Lambda)$.
By assumptions the set of such vectors~$f$ is dense in the
space~$H$.
Analogously, for the right hand side of~(\ref{eqn:aux15}) 
\[
\left(\Ad^{-1} + \Pi (\Pi^* B \Pi)^{-1} \Pi^*
    \right)f
    = \left(\Ad^{-1} + \Pi (\Pi^* B \Pi)^{-1} \Pi^*
    \right)(f_\perp  + B\Pi \varphi)
    = u
\]
where $u = \Ad^{-1}(f_\perp  + B\Pi \varphi) + \Pi \varphi$.
Application of $L_B$ defined in~(\ref{eqn:defLB})   to both sides gives the desired result
(here $h = \Pi \varphi$):
\begin{multline*}
 L_B \left(\Ad^{-1} + \Pi (\Pi^* B \Pi)^{-1} \Pi^*
    \right)f  = L_B  \left( \Ad^{-1}(f_\perp  + B\Pi \varphi) + \Pi
    \varphi\right)
\\
    =f_\perp  + B\Pi \varphi = f
\end{multline*}
The formula $L_B^{-1} =\Ad^{-1} + \Pi (\Pi^* B \Pi)^{-1} \Pi^*$ now
follows from the usual density arguments.


Consider operator-function~$Q(z) = Q_{\beta_0, \beta_1}(z)$ with
$\beta_0 = \Lambda + \Pi^* B\Pi$ and $\beta_1 = -I$.
\[
Q(z) = - (\overline{\beta_0 + \beta_1
M(z)})^{-1}  \beta_1 =
(\overline{\Lambda + \Pi^* B\Pi -  M(z)})^{-1} =
- (M_B(z))^{-1}
\]
Since $Q(0) = - (M_B(0))^{-1} =  - \Lambda_B^{-1} = (\Pi B
\Pi^*)^{-1}$ is bounded by assumption, ope\-ra\-tors~$Q_{\beta_0,
\beta_1}(z)$ exist and are bounded at least for small~$|z|$.
According to Theorem~\ref{thm:Extensions}
 and representation~(\ref{eqn:ResolventAbb}) (see Corollary~\ref{cor:Abb-1})
  the inverse $(\Abb)^{-1}$ is bounded  and
\[
    (\Abb)^{-1} = \Ad^{-1} + \Pi Q(0)\Pi^*
    = \Ad^{-1} - \Pi \Lambda_B^{-1}\Pi^*
    = \Ad^{-1} + \Pi (\Pi^* B \Pi)^{-1}\Pi^*
\]
which coincides with~$L_B^{-1}$.
The last assertion again follows from Theorem~\ref{thm:Extensions}. 
\qed
\end{proof}


Simple corollaries of
 Theorem~\ref{thm:LB_and_Abb} and definition of~$L_B$ are given below.
Their consequences will not be pursued here, see~\cite{Ryz1}, \cite{Ryz2} for further details.
\begin{rem}
   Statements of Theorem~\ref{thm:LB_and_Abb} can be used
   to describe
   dependence of M-operator on the particular choice
   of $\Lambda$ in Definition~\ref{defn:Gn} of
   boundary operator $\Gn$.
 Obviously, if $B = B^*$ and the operator~$\Gn$ is defined with~$\Lambda$ replaced by~$\Lambda_B$, i.~e.
 as ~$\Gn : \Ad^{-1} f + \Pi \varphi \mapsto \Pi^* f + \Lambda_B
\varphi$ for $f \in H$, $\varphi \in \Dom(\Lambda)$,
 then all results  remain valid with $L_B$ playing the role of operator~$\An$
 with $M_B(z)$ being the M\nobreakdash-function.
\end{rem}
\begin{rem}
 The equality $\Lambda = \Lambda_B$ 
 is only possible if
\[
 \Gn (I + \Ad^{-1}B )\Pi \varphi =0, \quad \varphi \in \Dom(\Lambda)
\]
by virtue of representations~$\Lambda = \Gn \Pi$ and~$\Lambda_B = - \Pi^* B\Pi = - \Gn \Ad^{-1} B\Pi$.
It does not follow that the solution to this equation is~$B =  - \Ad$.
In fact, this equality contradicts
  the assumption~$\Dom(B)\supset \Pi\Dom(\Lambda)$ about operator~$B$ since~$\Dom(\Ad) \cap \Ran(\Pi) = \{0\}$.
When $B$ is such that $\Lambda = \Lambda_B$  then~$L_B$ and $\An$ coincide due to Corollary~\ref{cor:aux1}
 (or the equalities $\beta_1 = - I$, $\beta_0 = 0$  that follow from Theorem~\ref{thm:LB_and_Abb}).

\end{rem}
\begin{rem}
    The operator~$A_K$ corresponding to $B = 0$ is an analogue of
    Krein's  extension of $A_{00}$ characterized by the
    boundary condition~$(\Gn - \Lambda\Gd)u  =0$
    see~\cite{AS}, \cite{Grubb}, \cite{Grubb2}, \cite{Krein2}.
 Note that the semiboundedness of~$A_{00}$ is not required for this definition of Krein's extension (cf.~\cite{HKS}).
 The formal equality~$B = \infty$ corresponds to the operator~$L_B = \Ad$.
\end{rem}


\section{Cayley Transform of M-function. Applications to the scattering theory}\label{CAYLEY}

This section outlines basic results on the scattering theory
 for operators corresponding to boundary conditions studied in the previous section.
In order to investigate selfadjoint and nonselfadjoint
  cases simultaneously
   the schema based on the functional model
    for linear operators suggested by S.~Naboko in papers~\cite{Na2}, \cite{Na3}
     is proposed.
The central ingredient of this schema
 is a dissipative operator whose
  B.~Sz.-Nagy and C.~Foias
  functional model~\cite{NagyFoias}
   serves as the model space for all
    operators under consideration.
 Papers~\cite{Na2}, \cite{Na3} are devoted to the case of additive perturbations of a selfadjoint operator
  and offer an explicit form for the ``model'' dissipative operator used in the study.
 The crucial element of the approach is availability of a certain factorization of the perturbation
 which allows the subsequent model construction.
These assumptions regarding perturbations render
 the schema of~\cite{Na2}, \cite{Na3} not applicable
 to linear operators associated with boundary value problems because
  they can not be represented as additive perturbations of one another.
 The obvious way to circumvent this difficulty, at least
  in case of elliptic boundary value problems,
   is to investigate inverse operators
    instead of the ``direct'' ones~\cite{BirScattering}.
Since the inverses are bounded,
 their differences are well defined bounded operators, and the method of~\cite{Na2}, \cite{Na3}
  is fully applicable provided the ``model'' dissipative operator is suitably chosen.
This chapter suggests such an operator
 based on considerations involving the Cayley transform of M\nobreakdash-function.
In addition, the required factorizations turn out to
  be direct consequences
   of formulas of previous section, see especially Corollary~\ref{cor:Abb-1}.
 Necessary connections
 to the theory of dissipative operators and functional models
 are established by
 the relationship between the M\nobreakdash-function
 and the so-called characteristic
 function of a ``minimal'' symmetric operator
 discussed in papers~\cite{DM1}, \cite{Kochub1}, \cite{Kochub2}, \cite{Straus}
 in other contexts.
 The exposition in this section is carried out
 in the spirit of
 nonselfadjoint operator theory
 and concludes with a brief sketch illustrating the proposed approach in  
 Remark~\ref{rem:Scattering}.
%
%


It is convenient to begin with the following observation.
Since values of $M(z)$,  $z\in \Complex_+$ are (possibly unbounded)
operators with positive imaginary part,
 operators $M(z) + i I$ are boundedly invertible
 for $z\in \Complex_+$.
 Moreover, a short argument shows that
 the Cayley transform of $M(z)$ defined as
 $\Theta(z) = (M(z) - iI)(M(z) + i I)^{-1}$
 is analytic and contractive for $z\in\Complex_+$.
 It turns out that $\Theta(z)$ for $z\in \Complex_+$
 is the characteristic function of some dissipative operator~$L$
 in the sense of A.~\v{S}traus~\cite{StrausCharF}.
 This fact was first observed in~\cite{Kochub1}, \cite{Kochub2}, \cite{Straus}
 for the characteristic function of Cayley transform of
 $A_{00}$ extended by the null map on $[\Ran(A_{00} + iI)]^\perp$ to the partial isometry defined everywhere in~$H$,
 and then reformulated in the setting of boundary triplets (under assumptions~$\Ran(\Gd) = \Ran(\Gn) = E$
 and  $M(z)$ bounded)
 for the characteristic function of respective dissipative operator in~\cite{DM1}.
The following Theorem renders these results in the form convenient for the discussion
 of nonselfadjoint scattering theory
  at the end of this section.
Note that boundedness of~$M(z)$ below is not required.
%

\begin{thm}\label{thm:CharFunct}
Operator~$L$ definded by the
 boundary condition
$(\Gn - i \Gd) u =0$ according to the Theorem~\ref{thm:Extensions}
 with $\beta_0 = -  iI$, $\beta_1 = I$
 is dissipative and boundedly
 invertible.
The inverse of $L$ is the operator $
 T = \Ad^{-1} - \Pi (\Lambda - i I )^{-1}\Pi^*.
$
For
 $z\in \Complex_+$
 the characteristic function of $L$ is given by the
 formula:
\[
 \Theta(z)   = (\Lambda -iI)(\Lambda + iI)^{-1}
 + 2 i z (\Lambda + iI)^{-1} \Pi^* (I -  zT^*)^{-1} \Pi (\Lambda + iI)^{-1}
\]
For $z\in \Complex_+$ this function
 coincides with the Cayley transform of~$M(z)$,
 \[
  \Theta(z) = (M(z) - iI)(M(z) + iI)^{-1}, \quad z \in  \Complex_+
 \]
\end{thm}

%


Before turning to the proof,
 let us recall the definition of characteristic function of~$L$
  according to~\cite{StrausCharF}.
 This definition is equivalent to the definition given by M.~Liv\v{s}i{c}~\cite{LivChar1}
 and independently by B.~Sz.-Nagy and C.~Foias~\cite{NagyFoias}
 and has been proven more convenient in practical applications.
%


%
Following~\cite{StrausCharF},
introduce a
 sesquilinear form~$\Psi (\cdot,
\cdot)$ defined on the do\-main $\Dom(L) \times \Dom(L)$:
\[
  \Psi(u, v) = \frac{1}{2i} [ (L u,v)_H - (u, L v)_H], \quad u, v \in \Dom(L)
\]
and a linear set~$
 \mathcal G(L) = \big\{ v \in \Dom(L) \mid \Psi (u,v) = 0,  \forall
 v
 \in \Dom(L) \big\}$.
Define the linear space~$\mathfrak L (L)$  as a closure of the
quotient space $\Dom(L) / \mathcal G(L)$ endowed with the inner
product $[\xi, \eta]_{\mathfrak{L}} = \Psi(f,g)$, where $\xi, \eta
\in \mathfrak L (L)$ and
 $u\in \xi$, $ v\in\eta$.
Obviously, $\mathfrak L(L) = \{0\}$ if $L$  is symmetric.
 A
\textit{boundary space}  for the operator $L$ is any linear space
$\mathfrak L$ which is isomorphic to $\mathfrak L(L)$.
A \textit{boundary operator} for the operator~$L$ is the linear
 map~$\Gamma$ with the domain~$\Dom(L)$ and the range in
 the boundary space~$\mathfrak L $ such that
\[
 [\Gamma u, \Gamma v ]_{\mathfrak L} = \Psi(u, v),
 \quad u, v \in \Dom (L)
\]
 Let $\mathfrak L'$ with the inner product~$[\cdot, \cdot]'$ be a boundary space for
 $-L^*$ with the boundary operator $\Gamma'$ mapping $\Dom(L^*)$ onto
 $\mathfrak L'$.
A \textit{characteristic function} of the operator~$L$ is an
operator-valued function~$\theta$ defined on the set~$\rho(L^*)$
whose values~$\theta(z)$ map $\mathfrak L$ into $\mathfrak L'$
according to the equality
 \[
\theta(z)\Gamma u = \Gamma'(L^* - zI)^{-1}(L-zI) u, \quad
u\in\Dom(L).
 \]
Since the right hand side of this formula is analytic with regard to
$z\in\rho(L^*)$, the function~$\theta$ is analytic on $\rho(L^*)$.


This construction needs to be applied
 to the  operator~$L$ from
 Theorem~\ref{thm:CharFunct} defined by the  boundary
 condition~$(\Gn - i \Gd)u =0 $.
 To that end, notice that $\beta_0 = -i I$ and  $\beta_1 = I$,
 and therefore $\mathscr B = \beta_0 + \beta_1 \Lambda  = -i I  + \Lambda$
 is boundedly invertible as $\Lambda = \Lambda^*$.
 The operator func\-ti\-on~$Q(z) = -(\overline{\beta_0 + \beta_1 M(z)})^{-1}\beta_1$
 has the representation~ $Q(z) = - (M(z) - iI)^{-1}$ and is bounded for ~$z\in
 \Complex_-$.
 In accordance with
 Theorem~\ref{thm:Extensions} and
 Corollary~\ref{cor:Abb-1} the inverse $L^{-1}$ exists and
 \[
  L^{-1} = \Ad^{-1} + \Pi Q(0)\Pi^* = \Ad^{-1} - \Pi (\Lambda -iI)^{-1}\Pi^*
 \]
 Denote $T = L^{-1}$ and compute the imaginary part of~$T$ defined as $\I (T)  = (T - T^*)/2i$.
 We have

  \begin{multline*}
  T - T^*   = L^{-1} - (L^{-1})^* =
   - \Pi (\Lambda -iI)^{-1}\Pi^* +
    \Pi (\Lambda  + iI)^{-1}\Pi^*
 \\
  =  \Pi (\Lambda  + iI)^{-1}
  \left[( \Lambda - iI) - (\Lambda + i I) \right] (\Lambda  - iI)^{-1} \Pi^*
\\
   = - 2i \Pi (\Lambda  + iI)^{-1} (\Lambda  - iI)^{-1} \Pi^*
  \end{multline*}
 Therefore
 \[
  \I (T) = \frac{T - T^*}{2i} = - \Pi (\Lambda  + iI)^{-1} (\Lambda  - iI)^{-1} \Pi^*
  \]
 which shows that  $T^*$ is  dissipative:
\begin{equation}\label{eqn:ImDiss}
 \I (T^*)  = \Pi (\Lambda  + iI)^{-1} (\Lambda  - iI)^{-1} \Pi^* \ge
0
\end{equation}


 The proof of Theorem~\ref{thm:CharFunct} is based on direct computations
 that closely follow the schema of A.~\v{S}traus~\cite{StrausCharF}.

\begin{proof}
 Suppose $u,v\in\Dom(L)$
  and denote $Lu = f$, $L v = g$.
 Then $f = T u$, $g = T v$ where $T = L^{-1}$
 and  for the form~$\Psi(\cdot;\cdot)$ we have
 \begin{multline*}
 \Psi(u,v)  = \frac{1}{2i}[(L u, v)-(u, L v)] =
 \frac{1}{ 2i} [(f,T g ) - (T f, g) ] 
\\
 =   \left(\frac{T^* - T}{2i} f, g\right)  =   \left(\I (T^*)f, g\right)
   = ((\Lambda  - i I)^{-1}\Pi^* f, (\Lambda  - i I)^{-1}\Pi^*  g)
\\
  = ((\Lambda  - i I)^{-1}\Pi^* L u, (\Lambda  - i I)^{-1}\Pi^* L v )
 \end{multline*}
Thus, the boundary space~$\mathfrak L$ for $L$ can be chosen as a
closure of $\Ran((\Lambda - iI)^{-1}\Pi^*)$ with the boundary
ope\-ra\-tor~$\Gamma = (\Lambda  - i I)^{-1}\Pi^* L  $:
\[
 \mathfrak L = \overline{ \Ran((\Lambda - i I)^{-1}\Pi^* L)}, \qquad
  \Gamma : u \mapsto (\Lambda - i I)^{-1}\Pi^* L u, \quad
  u \in \Dom(L)
\]
Note that the metric in~$\mathfrak L$ is positive definite,
 and $\mathfrak L$ is in fact a Hilbert space.
Analogous computations  for~$( - L^*) $ justify the following choice
of boundary space~$\mathfrak L'$ and boundary operator~$\Gamma'$
\[
 \mathfrak L' = \overline{ \Ran((\Lambda +iI)^{-1}\Pi^* L^*)}, \qquad
  \Gamma' : v \mapsto (\Lambda + i I)^{-1}\Pi^* L^* v, \quad
  v \in \Dom(L^*)
\]
Here $\mathfrak L'$ is a Hilbert space.
%
%


In order to calculate the characteristic function~$\Theta(z)$ of
 operator~$L$ corresponding to this choice of boundary spaces and
 operators, set again $u = Tf$ with $f \in H$
 so that $f = Lu$.
 For $z\in \rho(L^*)$ we have
 \[
  \begin{aligned}
   \Gamma' (L^* & - zI)^{-1} (L -z I)u   =
 (\Lambda + iI)^{-1}\Pi^* L^* (L^* - zI)^{-1} (L -zI)L^{-1}f
 \\
  &  = (\Lambda + iI)^{-1}\Pi^* (I - zT^*)^{-1} (I - zT)f
   \\
   & = (\Lambda + iI)^{-1}\Pi^* (I - zT^*)^{-1} (I - zT^* + z(T^* - T))f
   \\
   &= (\Lambda + iI)^{-1}\Pi^*\left[ I + 2i z (I - zT^*)^{-1}(\I (T^*))
   \right]f
   \\
   &  = (\Lambda + iI)^{-1}\Pi^*\left[ I + 2i z (I - zT^*)^{-1}\Pi(\Lambda + iI)^{-1}
 (\Lambda - iI)^{-1}\Pi^*f
   \right]
   \\
    & = \left[  (\Lambda - iI)(\Lambda + i I)^{-1} + 2i z
   (\Lambda + iI)^{-1}\Pi^* (I - zT^*)^{-1}\Pi(\Lambda + iI)^{-1}
    \right]  \times
\\ 
 &  \hspace{240pt}
\times (\Lambda - iI)^{-1}\Pi^* f
   \end{aligned}
 \]
Since $(\Lambda - iI)^{-1}\Pi^* f =  (\Lambda - iI)^{-1}\Pi^* L u =
 \Gamma u$, this formula shows that the characteristic function
 of $L$ coincides with the expression in brackets,
 that is, the function~$\Theta(z)$ from the Theorem statement.
%
%


 For the verification
 of identity~$\Theta = (M - iI)(M +
  iI)^{-1}$ write down the adjoint of function~$\Theta$
\[
 \left[\Theta(\bar z)\right]^* =
  (\Lambda + iI)(\Lambda - i I)^{-1}
   - 2i z
   (\Lambda - iI)^{-1}\Pi^* (I - zT)^{-1}\Pi(\Lambda - iI)^{-1},
   \quad z \in \Complex_-
\]
By virtue of equality~$Q(0) = - (\Lambda - iI)^{-1}$
 and Corollary~\ref{cor:Abb-1}
 \begin{multline*}
   z (\Lambda - iI)^{-1}\Pi^*
  (I - zT)^{-1}\Pi(\Lambda - iI)^{-1}
  \\
   =  z Q(0) \Pi^*  (I - z L^{-1})^{-1}\Pi Q(0)
  = Q(z) - Q(0)
   = - (M(z) - iI )^{-1} + (\Lambda - iI)^{-1}
 \end{multline*}
Therefore
\[
 \begin{aligned}
 \left[\Theta(\bar z)\right]^*   & =
   (\Lambda + iI)(\Lambda - i I)^{-1}
   + 2i (M(z) - iI )^{-1} -2i (\Lambda - iI)^{-1}
   \\
    & = I + 2i (M(z) - iI )^{-1} =
    (M(z) + iI) (M(z) - iI )^{-1}
 \end{aligned}
\]
By passing to the adjoint operators and noticing that $[M(\bar z)]^*
= M(z)$  the claimed identity follows.
\qed
\end{proof}


\begin{rem}
  The characteristic function
  of a linear operator is not determined uniquely~\cite{Brodskii}, \cite{NagyFoias}, \cite{StrausCharF}.
 Namely,
 consider two isometries $\tau : \mathfrak L \to \widetilde{\mathfrak L} $
 and $\tau' : \mathfrak L' \to \widetilde{\mathfrak L}' $ of the
  boundary spaces~$\mathfrak L  $, $\mathfrak L'  $
 of operator $L$
 to another
 pair of spaces~$\widetilde{\mathfrak L} $,
 $\widetilde{\mathfrak L}' $.
 It is easy to see that the characteristic function of $L$
 corresponding to the pair~$\widetilde{\mathfrak L}$,
 $\widetilde{\mathfrak L}' $ is the function~$\widetilde \theta(z) = \tau' \theta(z)\tau^* $.
In application to the characteristic function~$\Theta(z)$ above
 observe that the operator~$U = (\Lambda - iI)(\Lambda + iI)^{-1}$  is
 an unitary in~$E$.
 Therefore,
 both functions~$U^*\Theta(z)$ and $\Theta(z)U^*$, $z\in\Complex_+$
\[
 \begin{aligned}
  U^*\Theta(z) &  = I + 2i z(\Lambda - iI)^{-1} \Pi^* (I - zT^*)^{-1}
  \Pi (\Lambda + iI)^{-1}
  \\
  \Theta(z) U^* &  = I + 2i z(\Lambda + iI)^{-1} \Pi^* (I - zT^*)^{-1}
  \Pi (\Lambda - iI)^{-1}
 \end{aligned}
\]
 are characteristic
 functions of $L$, although corresponding to alternative
 choices of boundary spaces and operators.
\end{rem}


\begin{rem}
 Direct calculations according to the schema
outlined above
 yield the following expression  for the characteristic
  function~$\vartheta(z)$ of dissipative operator~$T^* =  (L^*)^{-1}$:
  \[
   \vartheta(\zeta) = I + 2i (\Lambda + iI)^{-1}\Pi^* (T - \zeta I )^{-1}
   \Pi(\Lambda -   iI)^{-1}, \quad \zeta \in\Complex_+
  \]
By virtue of~(\ref{eqn:ImDiss})
 this characteristic function is given in its ``standard'' form, which is consistent with
 the expression for characteristic function~$W(z) = I + 2i K^*(A^* - zI)^{-1}K$ of a bounded
 dissipative operator $A = R + i Q$
 with $R = R^*$, $Q = Q^* \geq 0$ and $Q=  K K^*$
 (or the corresponding operator node) that can be found in the literature~\cite{Brodskii}, \cite{NagyFoias}.
A close relationship between~$\vartheta(\zeta)$ and $\Theta(z)$ is clarified
by the
 substitution $\zeta\to z = 1/\zeta$
 \[
 \vartheta(1/z) = I - 2 i z
(\Lambda + iI)^{-1}\Pi^* (I - z T)^{-1}\Pi(\Lambda -
   iI)^{-1}, \quad z \in\Complex_-
 \]
Comparison with the expression for the adjoint of~$\left[\Theta(\bar
z)U^*\right]$ leads to
  the identity
  \[
   \vartheta(1/z) = U \left[ \Theta(\bar z)\right]^*, \quad z \in
   \Complex_-
  \]
where $U = (\Lambda - iI)(\Lambda + iI)^{-1}$ is an unitary.
\end{rem}


\begin{rem}\label{rem:Scattering}
  Dissipative operator~$T^*  = (L^{-1})^* = \Ad^{-1} - \Pi (\Lambda + iI)^{-1}\Pi^*$
  can be employed for the development of scattering theory
  of (in general, nonselfadjoint) operators~$L^\varkappa$
  defined by boundary
  conditions $(\Gn  + \varkappa\Gd) u = 0 $
  with $\varkappa : E \to E$.
 Assume $\Lambda + \varkappa$ is boundedly invertible.
Then the inverse~$T_\varkappa =  (L^\varkappa)^{-1}$ exists and
$T_\varkappa = \Ad^{-1} - \Pi (\Lambda + \varkappa)^{-1}\Pi^*$
 by Corollary~\ref{cor:Abb-1}.
 The functional model construction for additive perturbations~\cite{Na2}
   is fully applicable to~$\Ad^{-1}$, $T^*$, $T_\varkappa$,
    which makes possible development of the scattering theory for $\Ad^{-1}$ and $T_\varkappa$.
 Application of the invariance principle for the function~$t \to
  (1/t)$, $t\in \Real$,
   $t\neq 0$
   yields existence and completeness results for the local wave operators for the pairs
    $(\Ad,L^\varkappa)$, and $(L^\varkappa, \Ad)$.
 The
 interested reader is referred to the works~\cite{Na2}, \cite{Na3}, \cite{Ryz5}, \cite{Ryz3}
  for further details on the functional model of nonselfadjoint operators and
   its applications to the scattering theory.
\end{rem}

\section{Singular Perturbations}\label{EXAMPLE}
The schema developed in preceding sections is essentially axiomatic.
The only condition imposed on
 the set $\{\Ad^{-1}, \Pi, \Lambda\}$ is the validity of
  two Assumptions
   from Section~\ref{ABST}, whereas
     nothing specific is requested of the ``boundary''.
Due to this fact, our approach is applicable in situations not
 readily covered by the traditional boundary problems technique.
For instance, it makes possible a construction of ``boundary value
 problem'' when no boundary is given a priori.
 Introduction of an artificial boundary is a certain
  form of perturbation that is not ``regular'' in the
   traditional sense.
 Such ``singular'' perturbations are typical
  in the open systems theory  where they are
   identified with the open channels connecting the
    system with its environment~\cite{Liv}.
 From  this point of view,
 the selfadjoint operator~$\Ad$
 acting in the ``inner space''~$H$
 describes the ``unperturbed system''
 coupled with the  ''external space``~$E$
 by means of the ``channel'' operator~$\Pi : E \to H$.
The ``coupling'' takes place at the ``boundary''.
More details on connections to the
 open systems theory can be found in~\cite{Ryz2}.
%


%
%
%
 This section offers an
  illustration of these ideas
   by means of an elementary example 
   considered previously within the 
   framework of boundary triples in~\cite{Ryz4}. 
We study
 the physical model of a quantum particle
  in the potential field of finite number of singular
   interactions modeled by Dirac's $\delta$-functions.
The free particle is described by the
 Hamiltonian operator which in this case is the ``free'' Laplacian acting in
 $L^2(\Real^3)$, and the
 point interactions define ``perturbations'' of the unperturbed
 system (see~\cite{AFHKL}, \cite{AGHH}, \cite{AlbevKur} and
 references therein).
 Within the paper's context, the points where
 the interactions are situated
 form the ``boundary'' of the ``boundary value problem.''


Let~$H := L^2(\Real^3)$.
 Denote  $\Ad$ the selfadjoint boundedly
 invertible operator $I - \Delta $   in $H$
 with domain~$\Dom(\Ad) :=
H^2(\Real^3) $.
The fundamental solution to the equation~$((I -\Delta ) - zI) u =0$,
$z\in\Complex \setminus [1, \infty)$ is the square summable
function~$\mathscr G_z(x) = \frac{1}{4\pi}
\frac{\exp{(i\sqrt{z-1}|x|)}}{|x|}
$.
Fix a finite set of distinct points  $x_j \in\Real^3$, $j = 1,2,
\dots, n$ and introduce $n$ functions~$\mathcal
 G_j(x, z)  := \mathscr G_z(x - x_j)$.
Formally, each $\mathcal G_j(x, z)$ is the  solution to the partial
differential equation~$((I -\Delta ) - zI)u = \delta(x - x_j)$.
Any function~$\mathcal G_j(x, z)$ is infinitely differentiable in
any domain that does not contain~$x_j$.
Because of the singularity at~$x\to x_j$ functions~$\mathcal G_j(x,
z)$ are not in~$\Dom(\Ad)$.
However, for any $z,\zeta\in\Complex\setminus[1,\infty)$ the
difference~$\mathcal G_j(x, z) - \mathcal G_j(x, \zeta)$
 lies in~$\Dom(\Ad)$.
In the following the abridged notation~$\mathcal G_j$
for $\mathcal G_j(x, 0)$ will be used.
Notice that $\mathcal G_j$ are linearly independent as elements
of~$H = L_2(\Real^3)$.


 Choose the space~$E$ to be the $n$-dimensional
 Euclidian $E =\Complex^n$ with the orthonormal basis~$\{e_j\}_1^n$
 and define the operator~$\Pi : E\to H$ on $\{e_j\}_1^n$ by
 $\Pi : e_j \mapsto \mathcal G_j$.
 It follows that  $\Pi : a \mapsto \sum a_j \mathcal G_j$ where~$a = \sum a_j e_j$
 is an element of~$E$.
 Since $\Ran(\Pi)\cap \Dom(\Ad) = \{0\}$
 and the inverse to $\Pi$ is the mapping~$\sum a_j \mathcal G_j \mapsto \{ a_j\}_{j=1}^n$,
 Assumption~\ref{assum:1} holds.
Therefore
 we can introduce the operator~$A$ on
 domain~$\Dom(A) := \Dom(\Ad) \dot{+}{\mathcal H}$,
 where ${\mathcal H} := \Ran(\Pi) = \bigvee \mathcal G_j$.
 According to the Section~\ref{ABST},
 $A : \Ad^{-1} f + \sum a_j \mathcal G_j \mapsto f$, $f\ \in H$.
 The equality~$\Ker(A) = {\mathcal H}$ can be understood literally,
 because $(I -\Delta) \mathcal G_j = \delta(x-x_j) $ and the right
 hand side is supported on the set of zero Lebesgue measure in~$\Real^3$.
 Further,  the boundary operator~$\Gd$
 defined on~$\Dom(\Gd) = \Dom(A)$
 acts according to the rule $\Gd : f_0 + \sum a_j \mathcal G_j \mapsto \{a_j\}_1^n $,
 where $f_0 \in \Dom(\Ad)$ and $\{a_j\}_1^n\in E$.
 Due to  identity $\Gd \mathcal G_j = e_j$ we have
  $\Ker (\Gd) = \Dom(\Ad)$.
The requirements
 $\Gd\Pi =I_E$ and $\Pi\Gd \mathcal G_j = \mathcal G_j$ therefore are met.

%

 The operator~$S_z$ maps $a\in E $ into a unique solution~$u_z$
 of the equation $(A - zI) u =0 $ satisfying condition~$\Gd u = a$.
It is not difficult to see that $S_z$ has the form
\[
  S_z : \{a_j\}_1^n \mapsto u_z = \sum_j a_j \mathcal G_j(x, z), \quad z \in
  \Complex_\pm
\]
Indeed, the fact~$\mathcal G_j (x, z) \in \Ker(A -zI)$ was discussed
 above, and the boundary condition is verified by direct
 computations.
 For $a = \sum_j a_j e_j$ we have
\[
 \Gd S_z a = \sum_j a_j \Gd \mathcal G_j(x, z)
 = \sum_j a_j \Gd \mathcal G_j + \sum_j a_j \Gd \left(\mathcal G_j(x, z) - \mathcal G_j\right)
 = \sum_j a_j e_j = a
\]
because $\Gd \mathcal G_j = I $ and the difference $\mathcal G_j(x,
z) - \mathcal G_j$ belongs to $\Dom(\Ad)$,  therefore to
$\Ker(\Gd)$.


To calculate the adjoint~$\Pi^* : H \to E$ and choose
 the operator~$\Lambda $ in the representation $\Gn =
\Pi^* A + \Lambda \Gd$ appropriately
 suppose $a = \sum a_j e_j$ and $f\in H$.
 Then
 $(\Pi a,f) = \sum a_j(\mathcal G_j,f) = \langle a, \sum (f,\mathcal
 G_j)e_j\rangle$,
 hence~$\Pi^* $ is defined as
 $\Pi^* : f \mapsto \sum (f, \mathcal G_j)e_j$.
 If $f = \Ad f_0$ with some $f_0\in \Dom(\Ad)$,
 then $\Pi^* A f_0 = \Pi^* \Ad f_0 = \sum (\Ad f_0, \mathcal G_j)e_j$.
Summands here are easy to compute.
 It follows from the properties of fundamental
 solutions~$\mathcal G_j $ that $(\Ad f_0, \mathcal G_j) = f_0(x_j)$,
 therefore
 $\left.\Gn\right|_{\Dom(\Ad)} = \Pi^* A_0 : f_0 \mapsto \sum f_0(x_j)e_j$ for $f_0 \in \Dom(\Ad)$.
%


The  operator~$\Lambda$ describing $\Gn$ restricted to the
set~$\Ran(\Pi) $ can be chosen arbitrarily as long as it is
selfadjoint.
 For example, it could be taken as the identity $\Lambda = I_E$
 or the null operator $\Lambda : a\mapsto 0$, $a \in E$.
 However, it is convenient to define the action of $\Gn$ on $\Ran(\Pi)$
 consistently with its  action on $\Dom(\Ad)$.
 Since $\left.\Gn\right|_{\Dom(\Ad)}$ evaluates functions~$f_0 \in \Dom(\Ad)$
 at the points~$\{x_j\}_1^n$ and then builds a
 corresponding vector~$\{f_0(x_j)\}_1^n$ in~$E = \Complex^n$, we would
 like $\left.\Gn\right|_{\Ran(\Pi)}$ to act similarly.
 Functions $\mathcal G_j(x)$ are easily evaluated  at $x_s$ for $s\neq j$, but
 $\mathcal G_j $ is
 not defined at $x = x_j$; thus is not possible to
 define~$\Gn$ on $\Ran(\Pi) = \bigvee \mathcal G_j$
 to be the evaluation operator.
 To circumvent this problem recall that in the neighborhood of~$x_j$
 the function~$\mathscr G_z (x - x_j)$ has the following asymptotic expansion
 \begin{multline*}
 \mathscr G_z (x - x_j)  = \frac{1}{4\pi}\frac{   \exp{(i\sqrt{z-1}|x -x_j|)}  }{|x-x_j|}
\\ 
\sim
 \frac{1}{4\pi}\left(\frac{1}{|x-x_j|}  + i\sqrt{z-1} + O (|x -x_j|)\right)
  \end{multline*}
Define the action $\Gn$ on the vector~$\mathscr G_z(x - x_j)$ as
 \[
  \Gn : \mathscr G_z(x - x_j) \mapsto  \frac{i\sqrt{z-1}}{4\pi} e_j +
  \sum_{s\neq j} \mathscr G_z(x_j - x_s)e_s
 \]
 where $\frac{i\sqrt{z-1}}{4\pi}$ is the
 coefficient in the asymptotic expansion above corresponding to $|x-x_j|$ to the power
 $0$.
In particular, for $z =0$
\[
  \Gn : \mathcal G_j  \mapsto  -\frac{1}{4\pi} e_j +
  \sum_{s\neq j} \mathcal G_j\big|_{x = x_s}e_s
\]
where $ \mathcal G_j\big|_{x = x_s}  = \mathcal G_j (x_s, 0)=
\mathscr G_0(x_j - x_s)$, $s\neq j$.
 Thus for $a = \{a_j\}_1^n \in E$
 \[
 \Gn : \Pi a =
 \sum_j a_j \mathcal G_j \mapsto \Big\{- a_j \frac{1}{4\pi} + \sum_{s\neq j} a_s
 \mathcal G_s \big|_{x = x_j}
 \Big\}_{j = 1}^n
 \]


 The next step is the calculation of M-operator
 of $A$.
Quite analogously to the computation of $\Gn \Pi $ above we have for
$a = \{ a_j\}_1^n = \sum_j a_j e_j\in E$
\[
 \Gn : \sum_j a_j \mathcal G_j (x, z) \mapsto
 \Big \{a_j \frac{i\sqrt{z -1} }{4\pi} +
  \sum_{s\neq j}a_s \;\mathcal  G_s(x_j, z)
  \Big\}_{j =  1}^n
\]
 Since $S_z a = \sum_{j} a_j \mathcal G_j(x, z)$, this
 formula
 yields  for $M(z)a = \Gn S_z a$
\begin{multline*}
 M(z)a =  \Gn \Big( \sum_j a_j \mathcal G_j (x, z)\Big)
\\
 = \frac{1}{4\pi}\Big\{i a_j \sqrt{z-1}
 + \sum_{s\neq j} a_s \frac{ \exp{(i\sqrt{z-1}|x_j -x_s|)}}{|x_j -x_s|}
 \Big\}_{j =1}^n
\end{multline*}
Therefore the operator-function~$M(z)$  is the $n\times
n$\nobreakdash-matrix function with elements
\[
 M_{js}(z) = \frac{1}{4\pi}\left\{\quad
  \begin{aligned}
  i \sqrt{z -1}\; ,   & \quad j = s
  \\
  \frac{ \exp{(i\sqrt{z-1}\,|x_j -x_s|)}}{|x_j - x_s|}\; ,  & \quad j \neq s
  \end{aligned}
\right.
\]
%
%

%
%
%
%

By the change of variable $z\mapsto z + 1$ the matrix~$M(z+1)$ can
be interpreted as the M\nobreakdash-function
 of the Laplacian $-\Delta = A - I$
 in~$L_2(\Real^3)$ perturbed by a set of point
interactions~$\{\delta(x - x_j)\}_1^n$.
To elaborate more on this statement consider
  extensions of symmetric operator~$A_{00}$
 defined as $-\Delta + I$ on the domain
\[
 \Dom(A_{00}) =
 \{ u \in \Dom (\Ad)\mid \Gn u =0
 \} = \{ u \in H^2(\Real^3) \mid u(x_s) =0, \, s = 1,2,\dots n\}
 \]
Suppose
 the operator~$A^\beta$ is defined as a restriction of
 $A$ to
 domain~$\Dom(A^\beta) = \{ u \in \Dom(A)\mid (\beta_0\Gd + \beta_1\Gn)u =
 0\}$
where $\beta_0$, $\beta_1$ are arbitrary  $n\times
n$\nobreakdash-matrices.
The
 resolvent of~$A^\beta$
 is described in Theorem~\ref{thm:Extensions}.
 In particular, assuming that $\beta_0 + \beta_1\Lambda$ where $\Lambda = M(0)$
 is boundedly invertible, the
 inverse of $A^\beta$ as given by Corollary~\ref{cor:Abb-1} is
 \begin{equation}\label{eqn:Abeta-1}
   (A^\beta)^{-1} = \Ad^{-1} - \Pi (\beta_0 + \beta_1\Lambda)^{-1}\beta_1\Pi^*
 \end{equation}
Consider sesquilinear forms of both sides of this identity
 on a pair of vectors~$f,g\in H$.
Since~$\Ran(\Ad)  = H$, vectors~$f$ and $g$ can be
     represented as $f = \Ad
 u$,  $g = \Ad v$ with some $u, v \in \Dom(\Ad)$.
Then the form on the right is
\[
  (\Ad^{-1} f,g) - (\Pi (\beta_0 + \beta_1\Lambda)^{-1}\beta_1\Pi^*f,g)
   = (\Ad u,v) - ((\beta_0 + \beta_1\Lambda)^{-1}\beta_1\Gn u,\Gn v)
\]
 due to equalities $\Pi^* f = \Gn \Ad^{-1} \Ad u = \Gn u$
 and  $\Pi^* g = \Gn v$.
Notice that vectors~$\Gn u$ and $\Gn v$ are known explicitly, namely
$\Gn u = \{u(x_j)\}_1^n$ and $\Gn v = \{ v(x_j)\}_1^n$.

%
%
%
%

In order to clarify meaning of the form~$((A^\beta)^{-1}f,g)$ of
 the operator on
 the left hand side of~(\ref{eqn:Abeta-1})
 we need to recall some basic concepts from
 the theory of scales of Hilbert spaces~\cite{Berez}.
 Introduce the rigging~$H^+ \subset H
 \subset H^-$ of~$H$ constructed by the positive boundedly
 invertible operator ~$\Ad = -\Delta + I$.
The positive space~$H^+$ consists of elements from~$\Dom(\Ad)$
 and is equipped with the norm~$\|u\|_+ = \|\Ad u\|_H$, $u\in \Dom(\Ad)$.
It follows that~$\Ad$ acts as an isometry from $H^+$ onto $H$.
 The dual space~$H^-$ is
 identified with the Hilbert space
 of all antilinear functionals over elements from  $H^+$
 with respect to the inner product in $H$.
 In the usual way,
 the product~$(f,g)_H$ of two vectors $f,g \in H$
 is naturally extended to
 the duality relation
 between~$f \in H^-$ and $g \in H^+$.
This construction allows one to consider a continuation~$\Ad^+$
of~$\Ad$ from
 the domain~$\Dom(\Ad)$ to the whole of~$H$.
The map~$\Ad^+$ is defined on~$H$
 by the formula~$(\Ad^+ f, v ) = (f,\Ad v)$, $f\in H$,  $v\in H^+$
 and its range coincides with~$H^-$.
The sesquilinear form of~$(A^\beta)^{-1}$ on the left hand side
of~(\ref{eqn:Abeta-1}) calculated on the pair $\Ad u, \Ad v$ now can
be written as
\[
 ((A^\beta)^{-1}\Ad u,\Ad v) = (\Ad^+ (A^\beta)^{-1}\Ad u,v),
 \quad u, v\in H^+
\]
Thus the operator $\mathscr A^\beta := \Ad^+ (A^\beta)^{-1}\Ad $
acts from~$H^+$ into $H^-$ and its sesquilinear form is
\begin{equation}\label{eqn:DeltaInteractions}
  (\mathscr A^\beta u, v) = (u,v) + (-\Delta u, v )  +
  \sum_{j,k}\alpha_{jk}u(x_k)\overline{v(x_j)},
  \quad u,v\in H^2(\Real^3)
\end{equation}
where $\alpha_{jk}$ are the matrix elements of the operator~$-
 (\beta_0 + \beta_1\Lambda)^{-1}\beta_1$ in the basis~$\{e_j\}_1^n$.
%
%

%
%
%

Formula~(\ref{eqn:DeltaInteractions}) relates  ideas of
 this section to the conventional theory of point
 interactions.
It is easily seen that  the mapping~$L^\beta = \Ad^+ (A^\beta)^{-1}
\Ad $ is formally  represented as
 $-\Delta + I +  \alpha  (\cdot \, , \,\vec\delta) \, \vec \delta$
 where  $\vec \delta = \{\delta(x - x_j)\}_1^n$
 and ${\alpha}$ is the matrix~$\alpha = \|\alpha_{jk}\|$.
 Non-diagonal elements of~$\alpha$ describe pairwise
 interactions between points $\{x_j\}$ themselves
 (the so called ``non-local model''~\cite{KurPos}),
 whereas
 the standard case of $n$ mutually independent point
 interactions is recovered from~(\ref{eqn:DeltaInteractions})
 when the matrix~$\alpha$ is diagonal.
 Under assumption~$\beta_0 \beta_1^* = \beta_1 \beta_0^*$
  the operator~$A^\beta$ is selfadjoint according to
   Corollary~\ref{cor:AbbSelfadjoint}.
Finally,  Theorem~\ref{thm:discreteSpectrum} reduces the
 question of point
 spectrum of~$A^\beta$ to the study of
 $\det (\beta_0 + \beta_1 M(z) )$,
 where~$M(z)$ is the M\nobreakdash-function
 discussed above.
The point spectrum in the case~$\beta_1 =I$ and the matrix~$\beta_0$ diagonal
 was investigated in the work~\cite{Thom}.


Notice in conclusion that
 considerations of this section suggest a consistent way to construct singular
 perturbations of differential
 operators by ``potentials'' supported by sets of Lebesgue
 measure zero in $\Real^n$, cf.~\cite{AFHKL}.

%
%

\smallskip

%



\end{document}